\def\DRAFT{}

\ifdefined\EC
\documentclass[format=acmsmall, review=true]{acmart}
\usepackage{acm-ec-20}
\usepackage{booktabs} 
\usepackage[ruled]{algorithm2e} 

\SetAlFnt{\small}
\SetAlCapFnt{\small}
\SetAlCapNameFnt{\small}
\SetAlCapHSkip{0pt}
\IncMargin{-\parindent}
\setcitestyle{authoryear}
\author{Submission 278}
\else
\documentclass{article}
\usepackage{natbib}
\setcitestyle{authoryear,open={(},close={)}}
\fi

\usepackage{comment}
\usepackage{xcolor}
\usepackage{tikz}
\usepackage{graphics}
\usepackage{graphicx}
\usepackage{color}
\usepackage{gensymb} 
\usepackage{amsmath}
\usepackage{amsthm}
\usepackage{amsfonts}
\usepackage{geometry}
\usepackage[normalem]{ulem}
\usepackage{float}

\usepackage{hyperref}
\hypersetup{colorlinks=true,linkcolor=blue,citecolor=blue}
\usepackage{dsfont} 
\usepackage{nicefrac}

\ifdefined\EC
\title[Feasible Joint Posterior Beliefs]{Feasible Joint Posterior Beliefs}
\else
\title{Feasible Joint Posterior Beliefs\footnote{This paper greatly benefited from multiple suggestions and comments of our colleagues. We are  grateful  (in alphabetic order) to Kim Border, Ben Brooks, Laura Doval, Piotr Dworczak, Nikita Gladkov, Sergiu Hart, Kevin He, Aviad Heifetz, Yuval Heller,   Matthew Jackson, Eliott Lipnowski, Jeffrey Mensch, Benny Moldovanu,   In\'{e}s Moreno de Barreda, Stephen Morris,  Alexander Nesterov, Abraham Neyman, Michael Ostrovsky, Thomas Palfrey, Jim Pitman, Luciano Pomatto, Doron Ravid, Marco Scarsini, Eilon Solan, Theodore Zhu, Gabriel Ziegler, and seminar participants at Bar-Ilan Univeristy, Caltech, Hebrew University, HSE St.\,Petersburg, Technion, Tel Aviv University, Stanford, and UC San Diego.}}
\author{ \large Itai Arieli\thanks{Technion, Haifa (Israel). Itai Arieli is supported by the Ministry of Science and Technology (\#2028255).}\ \ \ %
Yakov Babichenko\thanks{Technion, Haifa (Israel). Yakov Babichenko is supported by a BSF award (\#2018397).} \ \ \ %
Fedor Sandomirskiy\thanks{Technion, Haifa (Israel) and Higher School of Economics, St.Petersburg (Russia).  Fedor Sandomirskiy is partially supported by the Lady Davis Foundation, by
	Grant 19-01-00762 of the Russian Foundation for Basic Research, by the European Research Council (ERC) under the European Union's Horizon 2020 research and innovation program (\#740435), and by the Basic Research Program of the National Research University Higher School of Economics. 	} \ \ \ %
Omer Tamuz\thanks{Caltech. Omer Tamuz is supported by a grant from the Simons Foundation (\#419427), by a BSF award (\#2018397), and by a
  Sloan Fellowship.}}
\fi

\ifdefined\DRAFT

\definecolor{ForestGreen}{rgb}{.13,.54,.13}
\definecolor{violet}{cmyk}{0.79,0.88,0,0}
\newcommand{\fed}[1]{{\color{ForestGreen}{(\textbf{Fedor:} #1)}}}

\newcommand{\yakov}[1]{{\color{red}{(\textbf{Yakov:} #1)}}}

\else
\newcommand{\fed}[1]{}
\newcommand{\yakov}[1]{}

\fi

\newtheorem{theorem}{Theorem}
\newtheorem{lemma}{Lemma}

\newtheorem{corollary}{Corollary}
\newtheorem{proposition}{Proposition}

\newtheorem{definition}{Definition}
\newtheorem{thmx}{Theorem}

\newtheorem{corx}{Corollary}

\theoremstyle{definition}

\theoremstyle{remark}

\newcommand{\E}{\mathbb{E}}
\newcommand{\bP}{\mathbb{P}}
\newcommand{\cP}{\mathcal{P}}

\newcommand{\p}{P}
\newcommand{\q}{Q}

\newcommand{\R}{\mathbb{R}}

\newcommand{\supp}{\mathrm{supp}\,}
\def\dd{\mathrm{d}}
\def\cL{\mathcal{L}}
\def\cF{\mathcal{F}}
\def\cG{\mathcal{G}}
\def\cA{\mathcal{A}}
\def\cB{\mathcal{B}}
\def\cC{\mathcal{C}}
\newcommand{\argmax}{\operatornamewithlimits{argmax}}

\newcommand{\half}{\nicefrac{1}{2}}

\begin{document}

\ifdefined\EC
\else
\maketitle
\fi

\begin{abstract}
We study the set of possible joint posterior belief distributions of a group of agents who share a common prior regarding a binary state, and who observe some information structure. For two agents we introduce a quantitative version of Aumann's Agreement Theorem, and show that it is equivalent to a characterization of feasible distributions due to \cite{dawid1995coherent}. For any number of agents, we characterize feasible distributions in terms of a ``no-trade'' condition. We use these characterizations to study information structures with independent posteriors. We also study persuasion problems with multiple receivers, exploring the extreme feasible distributions.


\end{abstract}

\ifdefined\EC
\maketitle
\fi



\section{Introduction}
The question of whether agents' observed behavior is compatible with rationality is fundamental and pervasive in microeconomics. A particularly interesting case is that of beliefs: when are agents' beliefs compatible with Bayes' Law?  

Consider a single agent who has a prior belief regarding an uncertain event, and who updates to a posterior belief after receiving some information. An analyst observes the prior and the distribution of posteriors, but not the information structure. When is the distribution of posterior beliefs feasible? That is, when is it compatible with Bayesian updating according to some information structure? Feasibility in this case corresponds to what is known as the martingale condition: the expectation of the posterior must equal the prior. This fundamental result is known as the Splitting Lemma \citep[see][]{blackwell1951comparison,aumann1995repeated}, and is a key tool in the theories of Bayesian persuasion \citep{kamenica2011bayesian} and of games with incomplete information \citep{aumann1995repeated}.

Thus, in the case of a single agent, the analyst can readily determine feasibility.\footnote{The question becomes more subtle when the agent is exposed to a collection of information sources. \cite*{brooks2019information} demonstrated that the \emph{feasible family} of  distributions that arises as we vary the subset of information sources does not admit a simple characterization even in the single agent case.} We ask the same question, but regarding groups of more than one agent. In particular, we consider a group of agents for which one can observe a common prior regarding a binary state, as well as a joint distribution of posteriors. This distribution is said to be feasible if it is the result of Bayesian updating induced by some joint information structure observed by the agents. Is there also in this case a simple characterization of feasibility, that does not require the analyst to test infinitely many possible information structures?

Clearly, feasibility implies that each agent's posterior distribution must satisfy the martingale condition. An additional important obstruction to feasibility is given by Aumann's seminal Agreement Theorem \citep{aumann1976agreeing}. Aumann showed that rational agents cannot agree to disagree on posteriors: when posteriors are common knowledge then they cannot be unequal. This implies that feasibility is precluded if, for example, agent 1 has posterior $1/5$ whenever agent 2 has posterior $2/5$, and likewise agent 2 has posterior $2/5$ whenever agent 1 has posterior $1/5$. There are, however, examples of distributions that are not feasible, even though they do not involve agents who agree to disagree as above. Thus the Agreement Theorem does not provide a necessary and sufficient condition for feasibility. Our first result is a quantitative version of the Agreement Theorem: a relaxation that constrains beliefs to be approximately equal when they are approximately common knowledge.

For the case of two agents, a characterization of feasible distributions was obtained by \cite*{dawid1995coherent}. Their criterion of feasibility takes a form of a family of inequalities, which arise from results due to  \cite{kellerer1961Funktionen} and \cite{strassen1965existence}. As we explain, these inequalities correspond exactly to those in our quantitative Agreement Theorem, which thus provides a necessary and sufficient condition for feasibility: a joint posterior belief distribution is feasible if and only if agents do not approximately agree to disagree.


For any number of agents, we find a characterization of feasibility via a ``no-trade'' condition. To derive this condition we introduce---as an analytic tool---a mediator who trades a good of uncertain value with the agents. When an agent trades, it is at a price that reflects her expected value, and thus the agents do not gain or lose in expectation. This implies that the mediator likewise does not turn a profit, and so by bounding the mediator's profit by zero we attain  an obviously necessary condition for feasibility. We show that this condition is also sufficient. Thus, for any number of agents, a joint posterior belief distribution is feasible if and only if no mediator can extract an expected profit from a trading scheme in which each agent expects to lose nothing.

We apply these characterizations to study {\em independent} joint posterior belief distributions: these are induced by information structures in which each agent receives non-trivial information regarding the state, and yet gains no information about the others' posterior. We give a simple condition for feasibility of independent distributions in the case of two agents with identically distributed, symmetric posteriors: such distributions are feasible if and only if the uniform distribution on $[0,1]$ is a mean preserving spread of the single agent belief distribution.

The set of feasible distributions plays an important role in Bayesian persuasion problems, and in particular in what we call {\em first-order} Bayesian persuasion. A first-order Bayesian persuasion problem includes a single sender and multiple receivers. First, the sender chooses an information structure to be revealed to the agents. Then, each receiver chooses an action. Finally, both the receivers and the sender receive payoffs. The sender's utility depends generally on the receivers' actions. The key assumption is that each receiver's utility depends only on the unknown state and her own action, and so only her first-order beliefs matter for her choice; we therefore refer to these problems as first-order persuasion problems. As each receiver's action depends only on her own posterior, and as the sender's utility depends on the actions of the receivers, the sender's problem is to choose an optimal feasible joint posterior belief distribution for the receivers.

As an example of a first-order Bayesian persuasion problem, we study the problem of a ``polarizing'' sender who wishes to maximize the difference between two receivers' posteriors. This example highlights the limitations imposed by feasibility, as the sender cannot hope to always achieve complete polarization in which one receiver has posterior 0 and the other posterior 1; such joint posterior distributions are precluded by the Agreement Theorem. We solve this problem in some particular regimes, and show that non-trivial solutions exist in others, depending on the details of how the difference between posteriors is quantified.

A related question is how anti-correlated feasible posteriors can be. While they can clearly be perfectly correlated, by the Agreement Theorem they cannot be perfectly anti-correlated. This question can be phrased as a first-order Bayesian persuasion problem, which we solve. When the two states are a priori equally likely, we show that the covariance between posteriors has to be at least $-1/32$, and construct a feasible distribution, supported on four points, that achieves this bound.\footnote{The questions of first-order persuasion for polarization and anti-correlation were mentioned by \citet[Problem 5.2]{burdzy2020bounds} as open problems.}


The question of first-order Bayesian persuasion is closely tied to the study of the extreme points of the set of feasible distributions, as the sender's optimum is always achieved at an extreme point. In the single agent case, the well-known concavification argument of \cite{aumann1995repeated} and \cite{kamenica2011bayesian} shows that every extreme point of the set of feasible posterior distributions has support of size at most two. In contrast, we show that for two or more agents there exist extreme points with countably infinite support. In the other direction, we show that every extreme point is supported on a set that is small, in the sense of having zero Lebesgue measure. To this end, we do not use our characterization of feasibility, but rather a classical result of \cite{lindenstrauss1965remark} regarding extreme points of the set of measures with given marginals. Likewise, our analysis of first-order persuasion is not based on the aforementioned characterizations of feasibility; to study these problems we exploit the fact that conditional expectations are orthogonal projections in the Hilbert space of square-integrable random variables.

This work leaves open a number of interesting questions. In particular, we leave for future studies the extension of our work beyond the setting of a binary state and common priors. Likewise, the study of the extreme feasible distributions is far from complete. For example, we do not have a simple characterization of extreme points. We also do not know if there are non-atomic extreme points.

\subsection*{Related literature}

 \paragraph{Coherent  experts' opinions.} The study of feasible joint distributions of posteriors was pioneered by \cite{dawid1995coherent}. They were motivated by the question of aggregating the opinions of experts who rely on different information sources. As a byproduct of their analysis, \cite{dawid1995coherent} characterized the set of feasible distributions for two agents and a binary state; the detailed discussion of this result is in \S\ref{sec:model}. Their foundational paper and the mathematical literature inspired by it seem to be  known little by the economic theory community.\footnote{We are indebted to Jim Pitman for introducing us to this literature.} We refer to \cite{burdzy2020bounds} and \cite{burdzy2019contradictory} for the references to the literature on experts, summary of the results known in the two-agent case (the main focus of this literature), and tight bounds on the probability that the pair of posteriors differ by more than $\delta$. Maximizing the latter probability can be seen as an example of a first-order Bayesian persuasion problem. 
Another example is offered by \cite{dubins1980maximal} who found the distribution maximizing the expected maximal posterior for any number of agents. A  particular case of this result for two agents follows from our analysis; see the discussion after Proposition~\ref{prop:quadratic-persuasion}.
Independently and concurrently \cite{cichomski2020maximal}  proved an analog of our Proposition~\ref{prop:quadratic-persuasion} using a Hilbert-space technique similar to ours. \cite{gutmann1991existence} discovered an analog of the characterization by \cite{dawid1995coherent} in the particular case of independent posteriors and demonstrated that the uniform distribution on the square is feasible; they did not discover the special role played by this distribution  (see our Proposition~\ref{thm:product-feasible}). 

\paragraph{Information design and necessary conditions for feasibility.} Recently, several necessary conditions for feasibility appeared in the economic literature studying information design with bounded-rational receivers.  Independently from us, \cite{ziegler2020adversarial} considered an information design setting with two receivers, and derived a necessary condition for feasibility. His condition is sufficient for the case of binary signals. We further discuss Ziegler's condition in \S\ref{sec:intervals}. Another necessary condition was found by \citet*[Proposition 4]{levy2018persuasion}; it 
is equivalent  to the corollary of Aumann's Agreement Theorem (Corollary~\ref{cor:agreement}) and, hence, is not sufficient. \cite{levy2018persuasion} also offered several recipes of how to construct feasible distributions starting from infeasible ones.

\paragraph{Alternative approaches to multi-agent information design.}
\cite*{mathevet2017information} studied Bayesian persuasion with multiple receivers and a finite number of signals. They found an implicit characterization of feasibility: considering the entire belief hierarchy, they showed that feasibility is equivalent to the consistency of the hierarchy.\footnote{Certain compatibility questions for belief hierarchies (without application to the feasibility of joint belief distribution or Bayesian persuasion) were recently addressed by \cite*{brooks2019information}.}.  \cite{bergemann2019information},  \cite{taneva2019information}, and \cite{arieli2019private} related the optimal information disclosure to the best Bayes Correlated Equilibrium from the sender's perspective. Even if receivers' actions are not free of externalities, finding the best such equilibrium leads to a linear program. This linear program happens to be tractable for a finite number of actions (\cite{taneva2019information} and \cite{arieli2019private} focused on the binary case). Our approach is conceptually closer to the geometric point of view on persuasion of \cite{kamenica2011bayesian}, where the distribution of posteriors plays the key role. Our approach helps ``visualize'' the solution and does not require the set of actions to be finite,\footnote{In the example of a polarizing sender that we discuss in \S\ref{sec:applications}, receivers have a continuum of actions (the set of actions coincides with the set of possible posterior beliefs), and the above linear-programming approach leads to an infinite-dimensional program.} however, it is limited to the first-order persuasion problems, i.e., we rule out strategic externalities.\footnote{In Section \ref{sec:applications} we provide several applications of our results and study feasible correlation of posterior beliefs. Related questions arise in numerous papers on information design. See, e.g., \cite{ely2017beeps} and \cite*{bergemann2020information}.}

\paragraph{ Common prior and no-trade.}
A related question to ours is the common prior problem studied in the literature on interactive epistemology. Concretely, an Aumann knowledge partition model with finite state space and agents is considered. Each agent has a partition over the state space and to each partition element corresponds a posterior probability that is supported on that partition element. The question is whether there exists a {\em common prior}: a single probability distribution over the state space that gives rise to all the posterior distributions by means of conditional probability.
\citet[Theorem 1a]{morris1992role} offered a characterization for the existence of a common prior in no-trade terms thus providing a variant of the converse statement to the no-trade theorem of \citet{milgrom1982information}. 
This result was rediscovered by \citet{feinberg2000characterizing} and a simple geometric proof was given  by \citet{samet1998common}.

There is a fundamental distinction between the common prior problem and ours. While in the common prior problem the conditional probability is given and therefore the full belief hierarchy at every state can be inferred, in our case only the unconditional posterior is considered, and the belief hierarchy is not specified.\footnote{In the notation introduced in our model section below, the common prior problem can be phrased as follows. Fix a signal space $S_i$ for each agent $i$, and denote $\Omega = \prod_i S_i$. Say we are given, for each agent $i$ and signal realization $s_i$, a conditional distribution $\mathbb{Q}_{s_i}(\cdot)$ supported on the subset of $\Omega$ in which agent $i$'s signal is $s_i$. When does there exist a single probability measure $\mathbb{P}$ on $\Omega$ such that for $\mathbb{P}$-almost every $s_i$ it holds that $\mathbb{Q}_{s_i}(\cdot) = \mathbb{P}(\cdot|s_i)$?} Despite this distinction, there is a connection between the no-trade characterizations of a common prior and of feasible  distributions. In a follow-up paper, \cite{morris2020notrade} 
demonstrated how to deduce a no-trade characterization of feasibility similar to ours from his earlier characterization of a common prior. This approach leads to a characterization of feasibility for finitely-supported distributions and  arbitrary finite sets of states.\footnote{The restriction to finitely-supported distributions can possibly be eliminated by approximation arguments.} \cite{morris2020notrade} also offered a comprehensive discussion of the history of the no-trade approach to the common prior problem. 



\paragraph{Measures with given marginals.}
From a technical perspective, our characterization of feasibility relies on the existence of measures with given marginals. Instead of the classic results of \cite{kellerer1961Funktionen} and \cite{strassen1965existence} used by \cite{dawid1995coherent}, we apply a more recent result due to \cite{hansel1986probleme}. In the economic literature, such tools were applied by \cite*{gershkov2013equivalence} and \cite*{gershkov2018theory}. Our feasibility condition for product distributions (Proposition~\ref{thm:product-feasible}) shares some similarity with Border's condition of feasibility for reduced-form auctions; see \cite*{hart2015implementation}.



\bigskip

The remainder of the paper is organized as follows. In \S\ref{sec:model} we introduce the formal problem. In \S\ref{sec:feasible} we present our main results. Applications are presented in \S\ref{sec:applications}.  In \S\ref{sec:persuasion} we study first-order Bayesian persuasion and the extreme feasible beliefs. In \S\ref{sec:implementation} we study implementations of a given, feasible posterior distribution.  Further proofs are provided in the Appendix.

\section{Model}
\label{sec:model}
\paragraph{Information structures and posterior beliefs.} 

We consider a binary state space $\Omega = \{\ell,h\}$, and a set of $n$ agents $N = \{1,2,\ldots,n\}$. An {\em information structure} $I=((S_i)_{i\in N},\bP)$ consists of signal sets $S_i$ (each equipped with a sigma-algebra, which we suppress) for each agent $i$, and a distribution $\bP \in \Delta(\Omega\times S_1\times \cdots\times S_n)$. Let $\omega,\,s_1,\ldots,s_n$ be the random variables corresponding to the $n+1$ coordinates of the underlying space $\Omega\times S_1\times \cdots \times S_n$. When it is unambiguous, we also use $s_i$ to denote a generic element of $S_i$. 

The \emph{prior probability} of the high state is denoted by $p= \bP(\omega=h).$ Throughout the paper we assume that $p\in(0,1)$. All $n$ agents initially have prior $p$ regarding the state $\omega$. Then, each agent $i$ observes the signal $s_i$. The \emph{posterior belief} $x_i$ attributed to the high state by agent $i$ after receiving the signal $s_i$ is
$$
  x_i = \bP(\omega=h\mid s_i).
$$
We denote by $\p_I$ the \emph{joint distribution of posterior beliefs} induced by the information structure $I$. This probability measure on $[0,1]^N$ is the joint distribution of $(x_1,x_2,\ldots,x_n)$. I.e., for each measurable $B \subset [0,1]^N$, 
$$
  \p_I(\mathrm{B})=\bP\Big((x_1,\ldots,x_n)\in \mathrm{B}\Big).
$$

We similarly denote the \emph{conditional joint distributions of posterior beliefs} by $\p_I^\ell$ and $\p_I^h$; these are the joint distributions of $(x_1,\ldots,x_n)$, conditioned on the state $\omega$.

For a probability measure $\p\in \Delta([0,1]^N)$ and for $i \in N$, we denote by $\p_i$ the marginal distribution, i.e., the distribution of the projection on the $i$\textsuperscript{th} coordinate, or the distribution of the posterior of agent $i$.

\paragraph{Feasible joint posterior beliefs.}  

When is a given distribution on $[0,1]^N$ equal to the distribution of posterior beliefs induced by some information structure? This is the main question we study in this paper. The following definition captures this notion formally.
\begin{definition}
  Given $p \in (0,1)$, we say that a distribution $\p \in \Delta([0,1]^N)$ is \emph{$p$-feasible} if there exists some information structure $I$ with prior $p$ such that $\p = \p_I$.
\end{definition}
When $p$ is understood from the context we will simply use the term ``feasible'' rather than $p$-feasible.

\paragraph{The single agent case.} 

Before tackling the question of feasibility for $n$ agents, it is instructive to review the well-understood case of a single agent. 

The so-called {\em martingale condition} states that the average posterior belief of a feasible posterior distribution is equal to the prior. Formally, given a $p$-feasible distribution $\p \in \Delta([0,1])$, it must be the case that
$$
  \int_0^1 x\, \dd\p(x)=p.
$$
The necessity of this condition for $p$-feasibility follows from the law of iterated expectation. For the single agent case, the martingale condition is necessary and sufficient for $\p$ to be $p$-feasible. This result is known as the Splitting Lemma.

\section{Feasible joint beliefs}
\label{sec:feasible}
\paragraph{The two agent case and the Agreement Theorem.} 

In the two agent case the martingale condition is not sufficient for feasibility. An additional obstruction to feasibility is given by Aumann's celebrated Agreement Theorem \citep{aumann1976agreeing}. We here provide a rephrasing of this theorem in a form that will be useful for us later.
\begin{thmx}[Aumann]
\label{thm:agreement}
Let $((S_i)_i,\bP)$ be an information structure, and let $B_1 \subseteq S_1$ and $B_2 \subseteq S_2$ be subsets of possible signal realizations for agents 1 and 2, respectively. If 
\begin{align}
    \label{eq:common-knowledge}
  \bP(s_1 \in B_1,\, s_2 \not \in B_2) = \bP(s_2 \in B_2,\, s_1 \not \in B_1) = 0
\end{align}
then
$$
  \E(x_1 \cdot \mathds{1}_{s_1 \in B_1}) = \E(x_2 \cdot \mathds{1}_{s_2 \in B_2}).
$$
\end{thmx}
To understand why this is a reformulation of the Agreement Theorem, note that condition \eqref{eq:common-knowledge} implies that 
$$
  \bP(s_1 \in B_1,\, s_2 \in B_2)  = \bP(s_1 \in B_1) = \bP(s_2 \in B_2),
$$  
and thus the event $\{s_1 \in B_1,\, s_2 \in B_2\}$ is self evident, i.e., is common knowledge whenever it occurs.\footnote{An event is self-evident if, whenever it occurs, all agents almost surely know that it has occurred. An event $A$ is common knowledge at an outcome $a\in A$ if $A$ contains a self-evident event that contains $a$.} 
 Hence this form of the Agreement Theorem states that if agents have common knowledge of the event $\{s_1 \in B_1,\, s_2 \in B_2\}$ then their average beliefs on this event must coincide. The original theorem follows by choosing $B_1$ and $B_2$ such that $x_1$ is constant on $B_1$ and $x_2$ is constant on $B_2$.  The statement of Theorem~\ref{thm:agreement} is close in form to that of the No-Trade Theorem of \cite{milgrom1982information}, which provides the same obstruction to feasibility; we discuss this further below. We do not provide a proof of Theorem~\ref{thm:agreement}, as it is a special case of our quantitative Agreement Theorem (Theorem~\ref{thm:agreement-quant}) which we prove below.

Using Theorem~\ref{thm:agreement} it is easy to construct examples of distributions that are not feasible, even though they satisfy the martingale condition. For example, consider the prior $p=\half$ and the distribution $\p=\frac{1}{2}\delta_{0,1}+\frac{1}{2}\delta_{1,0}$, where $\delta_{x_1,x_2}$ denotes the point mass at $(x_1,x_2)$. Clearly $\p$ satisfies the martingale condition. Nevertheless, it is not feasible, since under this distribution agents ``agree to disagree,'' when agent~1 attributes probability~$1$ to the high state and agent~2 considers the probability of this state to be $0$. Indeed, assume towards a contradiction that $\p$ is equal to $\p_I$, for some information structure $I = ((S_i)_i,\,\bP)$. Then the events $B_1 = \{s \in S_1\,:\, x_1(s) = 1\}$ and $B_2 = \{s \in S_2\,:\, x_2(s)=0\}$ satisfy the common knowledge condition \eqref{eq:common-knowledge}. But
$$
  \bP(\omega=h \mid s_1 \in B_1) = 1 \neq 0 = \bP(\omega = h \mid s_2 \in B_2),
$$
in contradiction to Theorem~\ref{thm:agreement}.

This example can be extended to a more general necessary condition for feasibility of a joint posterior distribution.
\begin{corx}
\label{cor:agreement}
Let $\p \in \Delta([0,1]^2)$ be feasible for some $p$, and let $A_1$ and $A_2$ be measurable subsets of $[0,1]$. Denote complements by $\overline{A_i} = [0,1] \setminus A_i$. If 
\begin{align}
  \label{eq:common-knowledge2}
  \p(A_1 \times \overline{A_2}) = \p(\overline{A_1} \times A_2) = 0
\end{align}
then
\begin{align}
    \label{eq:aumann-feasible}
  \int_{A_1} x\,\dd \p_1(x) = \int_{A_2} x\,\dd \p_2(x).
\end{align}
\end{corx}
Corollary~\ref{cor:agreement} is a recasting of Theorem~\ref{thm:agreement} into a direct condition for feasibility:  condition \eqref{eq:common-knowledge2} states that the event $A_1 \times A_2$ is self evident: it has the same probability as $A_1$ and the same probability as $A_2$. And each of the two integrals in \eqref{eq:aumann-feasible} is equal to the average belief of agent $i$ conditioned on $A_i$, times the probability of $A_i$. Hence Corollary~\ref{cor:agreement} follows immediately from Theorem~\ref{thm:agreement}, by setting $B_i = \{s_i \,:\, x_i(s_i) \in A_i\}$ for $i=1,2$. The advantage of this formulation is that it takes the form of a direct condition on $\p$.\footnote{A particular case of Corollary~\ref{cor:agreement} appears in \citet[Theorem 5.2]{dawid1995coherent}; and the inaccuracy in that statement was later corrected by \cite[Proposition 2.1]{burdzy2020bounds}. \cite{levy2018persuasion} found a necessary condition for feasibility equivalent to Corollary~\ref{cor:agreement} for distributions with finite support. None of these papers mention the connection to the Agreement Theorem.}


\paragraph{A quantitative Agreement Theorem.}
While Aumann's Agreement Theorem provides an obstruction to feasibility, a joint posterior distribution can be infeasible even when it does not imply that agents agree to disagree. A larger set of necessary conditions follows from our first  result, which is a quantitative version of the Agreement Theorem: 
\begin{theorem}
\label{thm:agreement-quant}
Let $((S_i)_i,\bP)$ be an information structure, and let $B_1 \subseteq S_1$ and $B_2 \subseteq S_2$ be sets of possible signal realizations for agents 1 and 2, respectively. Then
$$
  \bP(s_1 \in B_1,\, s_2 \not \in B_2) \geq \E(x_1 \cdot \mathds{1}_{s_1 \in B_1}) - \E(x_2\cdot\mathds{1}_{s_2 \in B_2})  \geq -\bP(s_2 \in B_2,\, s_1 \not \in B_1).
$$
\end{theorem}
\begin{proof}
By the law of total expectations, we have that
$$
\E(x_i\cdot \mathds{1}_{s_i\in B_i}) = \E(\bP(\omega=h|s_i)\cdot \mathds{1}_{s_i\in B_i}) = \bP(\omega=h,\,s_i \in B_i).
$$
We thus need to show that
$$
  \bP(s_1 \in B_1,\, s_2 \not \in B_2) \geq \bP(\omega=h,\, s_1 \in B_1) - \bP(\omega=h,\, s_2 \in B_2)  \geq -\bP(s_2 \in B_2,\, s_1 \not \in B_1).
$$
We show the first inequality; the second follows by an identical argument. We in fact demonstrate a stronger inequality:
\begin{equation}\label{eq:agreement_quant_stronger}
  \bP(\omega=h,s_1 \in B_1,\, s_2 \not \in B_2) \geq \bP(\omega=h,\, s_1 \in B_1) - \bP(\omega=h,\, s_2 \in B_2).
\end{equation}
Denote the conditional probability $\bP(C\mid \omega=h)$ by $\bP^h(C)$ for any event $C$. Then inequality~\eqref{eq:agreement_quant_stronger} is equivalent to the elementary inequality 
$\bP^h(A \cap \overline{B}) \geq \bP^h(A) - \bP^h(B)$, which holds for any pair of events $A$ an $B$ and any probability measure~$\bP^h$.
\end{proof}
A comparison to the Agreement Theorem (Theorem~\ref{thm:agreement}) is illustrative. In Theorem~\ref{thm:agreement} the common knowledge assumption \eqref{eq:common-knowledge} implies equality of average posteriors. Here \eqref{eq:common-knowledge} has been removed, and we instead bound the difference in the average posteriors by the extent to which \eqref{eq:common-knowledge} is violated. Thus, one can think of Theorem~\ref{thm:agreement-quant} as quantifying the extent to which approximate common knowledge implies approximate agreement.  The Agreement Theorem becomes the special case in which \eqref{eq:common-knowledge} holds.\footnote{An alternative approach to quantitative extensions of the Agreement Theorem is  by the concept of common $p$-beliefs \citep{monderer1989approximating}. \cite{neeman1996approximating} showed that when posteriors are common $p$-belief then they cannot differ by more that $1-p$. However, we are not aware of any formal connection between this extension and our Theorem~\ref{thm:two-agent-feasible}.}

In analogy to Corollary~\ref{cor:agreement}, we use Theorem~\ref{thm:agreement-quant} to derive further necessary conditions for feasibility.
\begin{corollary}
\label{cor:agreement-quant}
Let $\p \in \Delta([0,1]^2)$ be $p$-feasible for some $p$, and let $A_1$ and $A_2$ be measurable subsets of $[0,1]$. Then
\begin{align}
    \label{eq:aumann-feasible-quant}
  \p(A_1 \times \overline{A_2}) \geq \int_{A_1} x\,\dd \p_1(x) - \int_{A_2} x\,\dd \p_2(x) 
  \geq -\p(\overline{A_1} \times A_2).
\end{align}
\end{corollary}

Corollary~\ref{cor:agreement-quant} admits a simple interpretation in terms of the No-Trade Theorem \citep{milgrom1982information}. Consider three risk-neutral agents: two traders and a mediator. Trader $2$ owns a good with an unknown quality $\omega \in \{0,1\}$. The mediator also owns a copy of the same good. The two traders receive private information regarding the quality of the good, with a joint belief distribution $\p\in\Delta([0,1]^2)$. The mediator knows $\p$ and the realized pair $(x_1,x_2)$.

Let $A_1,A_2\subseteq [0,1]$ be any measurable sets and consider the following trading scheme: The mediator  buys the good from trader $2$ whenever $x_2\in A_2$ at a price of $x_2$. The mediator sells one copy of good to trader $1$ whenever $x_1\in A_1$ at a price of $x_1$. Thus the mediator may need to use her copy of the good, in case she sells but does not buy.

We argue that the mediator's expected profit is at least
$$
  \int_{A_1}x\, \dd\p_1(x)-\int_{A_2}x\,\dd\p_2(x)-\p(A_1\times \overline{A}_2).
$$     	
The first two addends correspond to the expected transfer between each trader and the mediator. The last addend corresponds to the event that the mediator has to sell his own good to trader $1$ since trader $2$'s belief $x_2$ is not in $A_2$ and trader $1$'s belief is in $A_1$. In this case the mediator loses at most 1.

Clearly, the mediator does not provide any additional information to the two players and so their expected profit is zero. Thus the mediator's expected profit is also zero, and so we have arrived at the left inequality of \eqref{eq:aumann-feasible-quant}. The right inequality follows by symmetry.

\paragraph{A characterization for two agents.}
\cite{dawid1995coherent} characterized the feasible distributions for the case of two agents, by applying a result of \cite{kellerer1961Funktionen} and \cite{strassen1965existence}. Although they do not relate their result to the Agreement Theorem or phrase it in these terms, what they show is that the condition of feasibility from Corollary~\ref{cor:agreement-quant} is both necessary and sufficient.
\begin{theorem}[\cite{dawid1995coherent}]
\label{thm:two-agent-feasible}
A probability measure $\p \in \Delta([0,1]^2)$ is $p$-feasible for some $p$ if and only if 
\begin{align}
    \label{eq:aumann-feasible-quant-2}
  \p(A_1 \times \overline{A_2}) \geq \int_{A_1} x\,\dd \p_1(x) - \int_{A_2} x\,\dd \p_2(x) 
  \geq -\p(\overline{A_1} \times A_2)\nonumber\\
  \text{for all measurable } A_1,A_2 \subseteq [0,1].
\end{align}
\end{theorem}
The necessity of \eqref{eq:aumann-feasible-quant-2} is Corollary~\ref{cor:agreement-quant}. The sufficiency requires another argument, which uses a theorem of \cite{kellerer1961Funktionen}. 
For the reader's convenience, we present the complete proof of Theorem~\ref{thm:two-agent-feasible} in Appendix~\ref{sec:proof}. 

It follows from Theorem~\ref{thm:two-agent-feasible} and the single agent martingale condition that when \eqref{eq:aumann-feasible-quant-2} holds then $\p$ is $p$-feasible for
\begin{align*}
    p = \int_0^1 x\,\dd \p_1(x) = \int_0^1 x\,\dd \p_2(x).
\end{align*}
We note that in \cite{dawid1995coherent},  condition~\eqref{eq:aumann-feasible-quant-2} was written in the equivalent form
$$\p(A_1 \times B_2)+p \geq \int_{A_1} x\,\dd \p_1(x) + \int_{B_2} x\,\dd \p_2(x),$$
from which the relation to the Agreement Theorem is harder to see.

It is natural to wonder if \eqref{eq:aumann-feasible-quant-2} can be relaxed to a simpler sufficient condition, and in particular if it suffices to check it on $A_1,A_2$ that are intervals. As we show in Appendix~\ref{sec:intervals}, restricting \eqref{eq:aumann-feasible-quant-2} to intervals results in a condition that is not sufficient: we construct a measure that is not feasible, but satisfies \eqref{eq:aumann-feasible-quant} for all intervals $A_1,A_2$. The constructed measure demonstrates that the condition derived by \cite{ziegler2020adversarial} for feasibility is necessary but insufficient.  

\paragraph{A characterization for any number of agents.}
For three or more agents, Theorem~\ref{thm:two-agent-feasible} provides a necessary condition for feasibility, as the joint belief distribution of each pair of agents must clearly satisfy \eqref{eq:aumann-feasible-quant-2}. However,  this condition is not sufficient: we construct below an example of three agents whose belief distribution satisfies \eqref{eq:aumann-feasible-quant-2} for each pair of agents, and yet is not feasible. The violation of feasibility stems from a violation of the No-Trade Theorem, in a manner similar to the one illustrated above for two agents. We use this approach to provide a necessary and sufficient condition for feasibility for an arbitrary number of agents.

A {\em trading scheme} consists of $n$ measurable functions $a_i \colon [0,1]\to[-1,1]$, $i=1,\ldots,n$. Given agent $i$'s posterior $x_i$, a mediator sells $a_i(x_i)$ units of the good to agent $i$ for the price of $x_i$ per unit, so that the total transfer is $a_i(x_i)x_i$. 

Clearly, each agent's expected profit is zero, since she is buying or selling the good at her expected value. Hence the mediator's expected profit is also zero. We argue that 
\begin{align}
    \label{eq:trading-1}
    \int_{[0,1]^n}\left(\sum_{i=1}^n a_i(x_i)x_i-\max\left\{0,\sum_{i=1}^n a_i(x_i)\right\}\right)\,\dd\p(x_1,\ldots,x_n) 
\end{align}
is a lower bound on the mediator's profit. Indeed, the first addend in the integral is the total transfer to the mediator. The second is equal to the total number of units of the good that the mediator needs to contribute to the transaction, in case the total amount that she sells exceeds the total amount that she buys. Since each unit is worth at most 1, she loses at most $\sum a_i(x_i)$, whenever this sum is positive. Thus, since the mediator's profit is zero, it follows that \eqref{eq:trading-1} cannot be positive if $\p$ is feasible.

Our characterization shows that this condition is also sufficient for feasibility. The proof is given in Appendix~\ref{app:no-trade}. 
\begin{theorem}
\label{thm:no-trade}
A probability measure $\p \in \Delta([0,1]^n)$ is $p$-feasible for some $p$ if and only if  for every trading scheme $(a_1,\ldots,a_n)$
\begin{align}
    \label{eq:trading}
    \int_{[0,1]^n}\left(\sum_{i=1}^n a_i(x_i)x_i-\max\left\{0,\sum_{i=1}^n a_i(x_i)\right\}\right)\,\dd\p(x_1,\ldots,x_n) \leq 0.
\end{align}
\end{theorem}
In the case of two agents, by taking $a_i=\pm\mathds{1}_{A_i}$ in~\eqref{eq:trading}, we recover the condition \eqref{eq:aumann-feasible-quant-2} from Theorem~\ref{thm:two-agent-feasible}. However, we are not aware of a simple argument for deducing \eqref{eq:trading} from \eqref{eq:aumann-feasible-quant-2} in the two agent case. In particular, Theorem~\ref{thm:two-agent-feasible} is not a simple corollary of Theorem~\ref{thm:no-trade} since the latter requires a broader set of trading schemes, while  indicators are enough for the former. Relatedly, Kellerer's theorem that underlies Theorem~\ref{thm:two-agent-feasible} holds only for $n=2$ and cannot be extended to the multidimensional case without expanding the set of test functions; see the discussion in \citet[pp. 436-437]{strassen1965existence}.

A natural question is whether Theorem~\ref{thm:no-trade} can be strengthened, along the lines of Theorem~\ref{thm:two-agent-feasible}, to consider only indicator trading schemes: Is it sufficient to consider trading schemes  of the form $a_i = \pm\mathds{1}_{A_i}$, so that  each agent is either a buyer or a seller, and has a set of beliefs in which she trades one unit? By computerized verification it is possible to show that the answer is no. A counterexample is the distribution $\nu \times \nu \times \nu \in \Delta([0,1]^3)$, where $\nu = \frac{1}{3}\left(\delta_{3/14}+\delta_{1/2}+\delta_{11/14}\right)$. This distribution is not feasible, and yet each of the small number of possible indicator trading schemes is not profitable. 

Since $\nu \times \nu$ is feasible, this example also shows that Theorem~\ref{thm:no-trade} provides additional obstructions for feasibility when $n\geq 3$, beyond the pairwise condition implied by Theorem~\ref{thm:two-agent-feasible}. We end this section with another such example. Consider three agents whose beliefs $(x_1,x_2,x_3)$ are distributed uniformly and independently on $[0,1]$. That is, let their joint belief distribution $\p \in \Delta([0,1]^3)$ be the Lebesgue measure. By Proposition~\ref{thm:product-feasible} below, the agents {\em pairwise} satisfy the condition of Theorem~\ref{thm:two-agent-feasible}. We argue that this is nevertheless not a feasible distribution. 

To see this, consider the trading scheme given by
\begin{align*}
    a_1(x)=a_2(x)=a_3(x) = \mathds{1}_{x \geq 2/3}-\mathds{1}_{x \leq 1/3}.
\end{align*}
In this scheme each agent buys a unit whenever her belief is above $2/3$, and sells when it is below $1/3$. A simple calculation shows that condition \eqref{eq:trading} of Theorem~\ref{thm:no-trade} is violated.\footnote{We get $\int_{[0,1]^3} a_i(x_i)x_i\dd \p= -\int_{0}^{\nicefrac{1}{3}} x\dd x +\int_{\nicefrac{2}{3}}^{1} x\dd x=\frac{2}{9}$. Hence, $\int_{[0,1]^3} \sum_{i=1}^3 a_i(x_i)x_i\dd \p=\frac{6}{9}$. The hyperplanes $x_i\in\{\nicefrac{1}{3},\nicefrac{2}{3}\}$ partition $[0,1]^3$ into $27$ small cubes.  There is $1$ small cube where the sum $\sum_{i=1}^3 a_i(x_i)$ is equal to $3$, there are $3$ cubes where the sum equals $2$, and $6$ cubes where it has the value of $1$; on the other cubes it is non-positive. Hence, $\int_{[0,1]^3} \max\left\{0,\,\sum_{i=1}^3 a_i(x_i)\right\}\dd \p=\frac{15}{27}=\frac{5}{9}$. We see that the condition \eqref{eq:trading} is violated and conclude that the uniform distribution on $[0,1]^3$ is infeasible.
We thank Eric Neyman for alerting us to an error in a previous version of this example.} These examples illustrate the general phenomenon that is captured  by Proposition~\ref{pro:product} below: for any distribution $\nu$ on $[0,1]$ not concentrated at one point, $\nu^n$ becomes infeasible for large enough $n$. In these two examples $n=3$ suffices.

\section{Applications}
\label{sec:applications}
\paragraph{Identically Distributed Binary Signals.}
As an illustration of the  restrictions imposed by the requirement of feasibility, consider
a setting with prior $p=\half$, and two agents who each receive a binary signal that equals the state with some probability $r > \half$. 
What joint distributions are feasible?

The canonical setting is the one in which signals are independent, conditioned on the state. Note that they are not unconditionally  independent: while each agent has each posterior with probability $\half$, conditioned on the first agent acquiring a high posterior, the second agent is more likely to also have a high posterior than a low one.  In this case the induced belief distribution is
$$
  \p = \frac{r^2+(1-r)^2}{2}\left[\delta_{r,r}+\delta_{1-r,1-r}\right]+r(1-r)\left[\delta_{1-r,r}+\delta_{r,1-r}\right].
$$  
Another simple case is the one in which both agents observe the same signal, in which case the posteriors are of course perfectly correlated, and the distribution is
$$
  \p = \frac{1}{2}\delta_{r,r}+\frac{1}{2}\delta_{1-r,1-r},
$$  
as in the signal agent case. In both of these cases $\p$ is feasible, since we derive it from a joint signal distribution.
 
The case in which agents' posteriors are perfectly anti-correlated, i.e.,
$$
  \p = \frac{1}{2}\delta_{r,1-r}+\frac{1}{2}\delta_{1-r,r}
$$  
is precluded by the Agreement Theorem and its Corollary~\ref{cor:agreement}, as agents here agree to disagree on their posteriors.

More generally, we can consider the case in which conditioned on an agent's posterior, the other agent has the same posterior with probability $c$. That is,
\begin{align}
\label{eq:binary}
  \p = \frac{c}{2}\left[\delta_{r,r}+\delta_{1-r,1-r}\right]+\frac{1-c}{2}\left[\delta_{1-r,r}+\delta_{r,1-r}\right].
\end{align}
In this case there is never common knowledge of beliefs, as long as $c<1$. The natural question is: to what extent can the signals be anti-correlated? Can they for example be (unconditionally) independent, so that after observing a signal, the probability that an agent assigns to the other agent's posterior is still uniform over $\{r,\,1-r\}$? A common intuition suggests that this is impossible, since  even if signals are independent given the state, the induced unconditional distribution of posteriors inherits the dependence on the state and thus the posteriors must be dependent; on the other hand, if the conditional distribution of signals is the same in both states, they convey no information and thus the posterior just equals the prior. Perhaps surprisingly, this intuition is wrong, and posteriors can be independent and even anti-correlated; see e.g.\ \cite{burdzy2020bounds}.
\begin{proposition}
\label{prop:binary}
A joint belief distribution $\p$ as given in \eqref{eq:binary} is $\half$-feasible if and only if $c \geq 2r-1$.
\end{proposition}
This proposition follows directly from Theorem~\ref{thm:two-agent-feasible}: one only needs to check the sets $A_i = \{r\}$ and $A_i = \{1-r\}$. More generally, for finitely supported $\p$ only finitely many conditions need be checked.

Proposition~\ref{prop:binary} implies that indeed too much anti-correlation is infeasible, especially as the signals become more informative. Nevertheless, it is possible that the agents' posteriors are independent of each other (i.e., $c = \half$) as long as $r \leq 3/4$. Moreover, for $r < 3/4$, the posteriors can be negatively correlated; for example, $\p$ is feasible for $c=1/3$ and $r=2/3$. In this case, posteriors are either $1/3$ or $2/3$, each obtained with probability $\half$. When an agent has the high posterior, she assigns probability $2/3$ to the event the the other agent has the low posterior.

\paragraph{Unconditionally independent signals.}
As another application of Theorem~\ref{thm:two-agent-feasible} we further explore independent joint posterior belief distributions. To motivate this application consider a sender (e.g., a consulting agency) who wants to reveal some information to receivers (its clients). However, there is an additional concern: none of the receivers must be able to form any non-trivial guess about the information received by their counterpart. This can be motivated either by privacy concerns or by the desire to avoid complicated strategic reasoning on receivers' side. For example, consider the case that the receivers are two firms competing on the same market and plan to use the received information to adjust their strategies. If their posteriors are not independent, they might engage in a complicated reasoning involving  higher order beliefs, as in \cite{weinstein2007impact}. Another motivation for studying independent joint beliefs comes from mechanism design, where these distributions arise endogenously \citep*[see, e.g.,][]{bergemann2017first, brooks2019optimal}.


We already saw above that identically distributed binary signals can be independent for prior $p=\half$ as long as $r \leq 3/4$. As a second example, let $\p$ be the uniform distribution on the unit square. Following \cite{gutmann1991existence}, we verify that it is $\half$-feasible and find the information structure inducing it. This distribution clearly satisfies the martingale condition so it remains to check that
$$
 \p(A_1 \times \overline{A_2}) \geq \int_{A_1} x\,\dd \p_1(x) - \int_{A_2} x\,\dd \p_2(x).
$$
The other inequality of \eqref{eq:aumann-feasible-quant-2} will follow by symmetry.

Let $1-a$ be the Lebesgue measure of $A_1$ and $b$ of $A_2$. Then the left hand side equals $(1-a)(1-b)$ and the right hand-side is maximized when $A_1$ is pushed as much as possible towards high values of the integrand  (i.e., $A_1=\left[a ,1\right]$) and $A_2$ is pushed towards low values ($A_2=\left[0,b\right]$). We get the following inequality
$$
  (1-a)(1-b) \geq\int_{a}^1 x \,\dd x-\int_0^{b}x\,\dd x= \frac{1}{2}\left(1-a^2-b^2\right).
$$
Simplifying the expression, we get $(a+b-1)^2\geq 0$, which holds for any $a$ and $b$. Thus the uniform  distribution is $\half$-feasible.

The equality attained at $a+b-1=0$ helps guess the information structure that induces this distribution of beliefs. Comparing the inequality in the statement of Theorem~\ref{thm:agreement-quant} to inequality~\eqref{eq:agreement_quant_stronger}, we see that for any such information structure, $\bP(\omega=h, x_1\in [a,1], x_2\in [b,1])=\bP(x_1\in [a,1], x_2\in [b,1])$ for $a+b-1=0$, and hence $\bP(\omega=h\mid x_1\in [a,1], x_2\in [b,1])=1$. In other words, whenever the pair of posteriors is in the triangle $T=\{(x_1,x_2)\in[0,1]^2\,:\, x_1+x_2> 1\}$, the state is $\omega=h$; by the symmetric argument, $\omega=\ell$ whenever posteriors are in $[0,1]^2\setminus T$. Hence, one can use the following information structure:   sample a pair of signals $(s_1,s_2)$ uniformly from $T$ if $\omega=h$ and from $[0,1]^2 \setminus T$ if $\omega=\ell$. This information structure leads to $x_i(s_i)=s_i$ and thus induces the uniform distribution.
In \S\ref{sec:implementation}, we describe a family of information structures inducing the uniform distribution on $[c,1-c]^2$.

We generalize this example to study the conditions for $\half$-feasibility of more general product distributions. The next proposition provides a necessary and sufficient condition for feasibility of a large class of such distributions, and shows that the uniform distribution is, in fact, an important edge case. We say that a distribution $\nu \in \Delta([0,1])$ is symmetric around $\half$ if its cumulative distribution function $F(a) = \nu([0,a])$ satisfies $F(a)=1-F(1-a)$ for all $a \in [0,1]$. Recall that $\mu \in \Delta([0,1])$ is a {\em mean preserving spread} of $\mu' \in \Delta([0,1])$ if there exist random variables $x,x'$ such that $x \sim \mu$, $x' \sim \mu'$ and $\mathbb{E}(x|x') = x'$.\footnote{Equivalently, $\int f(x)\,\dd\mu(x) \geq \int f(x)\,\dd\mu'(x)$ for every bounded convex $f$. Another equivalent condition is that both have the same expectation, and $\int_0^y F(x)\,\dd x \leq \int_0^y F'(x)\,\dd x$ for all $y \in [0,1]$, where $F$ and $F'$ are the cumulative distribution functions of $\mu$ and $\mu'$, respectively. See \citep{blackwell1953equivalent}.}
\begin{proposition}
\label{thm:product-feasible}
Let $\p=\nu\times \nu$, where $\nu\in \Delta([0,1])$ is symmetric
around $\half$. Then $\p$ is $\half$-feasible if and only if the uniform distribution on $[0,1]$ is a mean preserving spread of $\nu$.
\end{proposition}
In particular, among symmetric, $\half$-feasible product distributions, the uniform is maximal in the convex order.\footnote{We note that the  symmetry assumption cannot be dropped. Indeed, the uniform distribution is not a mean-preserving spread of the non-symmetric distribution $\nu=a\cdot \delta_{1-a}+(1-a)\cdot \delta_1$ with $a=\frac{1}{\sqrt{2}}$, but nevertheless $\nu\times\nu$ is feasible. Feasibility can checked via Theorem~\ref{thm:two-agent-feasible} or directly by constructing the information structure: both agents have two signals $S_1=S_2=\{\underline{s},\overline{s}\}$; if the state is $\omega=h$, then the pair of signals is $(\overline{s},\overline{s})$ with probability $2(1-a)^2$ and $(\underline{s},\overline{s})$ or $(\overline{s},\underline{s})$ with probabilities $\frac{1}{2}-(1-a)^2$; if the state is $\omega=\ell$, the signals are always $(\underline{s},\underline{s})$. The uniform distribution is not a mean-preserving spread of $\nu$ since $\int_{\frac{1}{2}}^1 \left(x-\frac{1}{2}\right) \dd \nu(x)=\frac{1}{2}-\frac{1}{2\sqrt{2}}> \int_{\frac{1}{2}}^1 \left(x-\frac{1}{2}\right) \dd x=\frac{1}{8}.$}

The proof of Proposition~\ref{thm:product-feasible} is relegated to \S\ref{sec:proofs}. The ``if'' direction is a consequence of the following, more general lemma.
\begin{lemma}
\label{lem:mps-monotone}
Let $\p = \mu_1 \times \cdots \times \mu_n \in \Delta([0,1]^n)$ be $p$-feasible, and let $\p' = \mu_1' \times \cdots \times \mu_n'$, where  each $\mu_i$ is a mean preserving spread of $\mu_i'$. Then $\p'$ is also $p$-feasible.
\end{lemma}
This lemma can be seen as a corollary of Blackwell's Theorem \citep[Theorem 12]{blackwell1951comparison}: Given an information structure that induces $\p$, one can apply an independent garbling to each coordinate to arrive at a structure that induces $\p'$. Our proof illustrates how the result can be obtained via the no-trade arguments of our Theorem~\ref{thm:no-trade}, without invoking Blackwell's Theorem.


Using Proposition~\ref{thm:product-feasible}, we show in~\S\ref{sect_Gaussian}  that Gaussian signals may induce independent beliefs, provided that that they are not too informative. Let $\nu$ be the belief distribution induced by a signal which, conditioned on the state, is a unit variance Gaussian with mean $\pm d$. Then $\nu\times \nu$ is $\half$-feasible if and only if $d$ lies between the $\frac{1}{4}$-quantile and $\frac{3}{4}$-quantile of the standard normal distribution.

\medskip

Interestingly, product distributions with a given marginal cease being feasible once the number of agents becomes large enough, as demonstrated in the following proposition.

 \begin{proposition}\label{pro:product}
    For every probability measure $\nu \in \Delta([0,1])$ that differs from a Dirac measure, for sufficiently large $n$ the product distribution $\nu^n \in \Delta([0,1]^n)$ is not feasible.
 \end{proposition}
 
 The proof, that is relegated to \S\ref{sec:proofs}, uses our Theorem~\ref{thm:no-trade}. We show that with sufficiently large number of agents, a mediator can implement a strictly beneficial trading scheme, i.e.\ one that violates Theorem~\ref{thm:no-trade}. In fact, we show that the product distribution $\nu^n$ is infeasible whenever $\lfloor \frac{n}{2} \rfloor > \frac{1}{8}\big(\int_m^1 x\,\dd\nu(x)-\int_0^m x \,\dd\nu(x)\big)^{-2}$, where $m$ is the median of the marginal distribution $\nu$.
 

\section{First-order Bayesian Persuasion and Extreme Feasible Joint Posterior Belief Distributions}
\label{sec:persuasion}
\paragraph{First-order Bayesian persuasion.}
In this section we consider a sender who sends information regarding an underlying state to a group of $n$ receivers. The sender's utility depends on the actions of the receivers. We assume that each receiver's utility depends only on the state and her own action, as, for example, is common in the social learning literature. Since the equilibrium action of a receiver is dictated solely by her first-order beliefs, we call this setting {\em first-order} Bayesian persuasion.  We note that the results of this section do not rely on the characterizations of feasibility given in Theorems~\ref{thm:two-agent-feasible} and~\ref{thm:no-trade}, but rather offer an additional set of tools to study the set of feasible distributions of posteriors. 

Formally, a first-order Bayesian persuasion problem is given by $B = (N,\,p,\,(A_i)_{i \in N},\,(u_i)_{i \in N},\,u_s)$. As above, $\omega \in \{\ell,\,h\}$ is a binary state for which $n = |N|$ receivers have a common prior $p \in (0,1)$. Each receiver $i \in N$ has to choose an action $a_i$ from a compact metric set of actions $A_i$. Her utility  $u_i(\omega, a_i)$ depends only on the state and her action. A single sender has utility $u_s(a_1,\ldots,a_n)$ which depends on the receivers' actions. 

The sender chooses an information structure $I = ((S_i)_{i \in N},\,\bP)$ with prior $p$, and each receiver $i$ observes a private signal $s_i \in S_i$ and then chooses an action $a_i$. 

In equilibrium, the action
$$
  \tilde a_i \in \argmax_{a_i \in A_i}\E(u_i(\omega,a_i)\mid s_i)
$$ 
is chosen by $i$ to maximize her expected utility conditioned on her information $s_i$. Note that since $u_i$ depends only on $\omega$ and $a_i$, it follows that given the information structure $I$, receiver $i$'s posterior $x_i = \bP(\omega=h\mid s_i)$ is a sufficient statistic for her utility, i.e.,
$$
  \max_{a_i \in A_i}\E(u_i(\omega,a_i)\mid s_i) = \max_{a_i \in A_i}\E(u_i(\omega,a_i)\mid x_i),
$$
and so a receiver does not decrease her expected utility by discarding her private signal, retaining only the first-order posterior belief $x_i$. We accordingly consider only equilibria in which the agents---even when they are indifferent---use only $x_i$ to choose their actions, so that $\tilde a_i$
is a function of $x_i$.\footnote{This refinement rules out equilibria in which, for example, an agent uses higher order beliefs to break ties when indifferent between two actions.}

The information structure $I$ is chosen to maximize the expectation of $u_s$. We assume that $u_i$ and $u_s$ are upper-semicontinuous, to ensure the existence of equilibria. The value $V(B)$ is the sender's expected equilibrium utility in the first-order Bayesian persuasion problem $B$.

A crucial feature is the assumption that $u_i$ does not depend on the other agents' actions. The case in which externalities are allowed is the general problem of Bayesian mechanism design, which is beyond the scope of this paper.\footnote{Nevertheless, our approach remains  applicable even in the presence of externalities if the game played by receivers is ``simple'' in the sense  of \cite{borgers2019strategically}, i.e., the equilibrium behavior of a receiver does not require her to form higher order beliefs.} In contrast, in first-order Bayesian persuasion the receivers have no strategic incentives. This implies that their higher order beliefs are irrelevant to the sender, who in turn is solely interested in their first-order posterior beliefs. This is captured by the following proposition, which states that for every first-order Bayesian persuasion problem there is an indirect utility function that the sender maximizes by choosing the receivers' posteriors. Of course, the posterior distribution must be feasible. 

Denote by $\cP_p^N$ the set of $p$-feasible $\p \in \Delta([0,1]^N)$.
\begin{proposition}
\label{prop:social-persuasion}
For every first-order Bayesian persuasion problem $B = (N,\,p,\,(A_i)_i,\,(u_i)_i,\,u_s)$ there is an {\em indirect utility function} $v \colon [0,1]^N \to \R$ such that the value $V(B)$ is given by
$$
  V(B) = \max_{\p \in \cP_p^n}\int v(x_1,\ldots,x_n)\,\dd \p(x_1,\ldots,x_n).
$$  
\end{proposition}

\paragraph{Example: Polarizing receivers.}
In \S\ref{sec:applications} we showed that joint belief distributions can be anti-correlated, and explored the extent to which identically distributed binary signals can be anti-correlated. For general signal distributions, this question can be formalized using a first-order Bayesian persuasion approach. 

Consider a sender whose indirect utility for the posteriors of two receivers is
$$
  v(x_1,x_2) = |x_1-x_2|^a,
$$
for some parameter $a > 0$. Informally, the sender wishes to maximize the polarization between receivers, or the discrepancy between their posteriors.

For the case $a=2$ we solve the sender's problem completely.
\begin{proposition}
\label{prop:quadratic-persuasion}
Let $V$ be the value of the two receiver first-order Bayesian persuasion problem where the sender's indirect utility is
$$
  v(x_1,x_2) = (x_1-x_2)^2
$$
and the prior is $p \in (0,1)$. Then $V(B) = (1-p)p$.
\end{proposition}
In this case of $a=2$, the optimum can be achieved by completely informing one agent, and leaving the other completely uninformed. E.g., by letting $s_1 = \omega$ and $s_2 = 0$. We show in \S\ref{sec:proofs} that when $p=\half$, the same policy is also optimal for all $a < 2$. \cite{dubins1980maximal} show the same result for $a=1$ and any $p$.

For $a\geq 3$, it is no longer true that it is optimal to reveal the state to one receiver and leave the other uninformed, which yields utility $2^{-a}$ to the sender. For example, the posterior distribution
\begin{align}
    \label{eq:extreme-3}
    \p = \frac{1}{4}\delta_{0,2/3}+\frac{1}{4}\delta_{2/3,0}+\frac{1}{2}\delta_{2/3,2/3}
\end{align}
can be easily verified to be $\half$-feasible using \eqref{eq:aumann-feasible-quant-2}, and yields utility
$\frac{1}{2}\left(\frac{2}{3}\right)^a$, which for $a=3$ (for example) is larger than $2^{-a}$.

\medskip
Another approach to quantifying the extent to which beliefs can be anti-correlated is by directly minimizing their covariance. This question admits a simple, non-trivial solution, as the next proposition shows: for the prior $p=\half$, the smallest possible covariance between posterior beliefs is $-1/32$, and is achieved on a distribution supported on four points.
\begin{proposition}
\label{prop:anti}
Let $V$ be the value of the two receiver first-order Bayesian persuasion problem where the prior is $p=\half$ and the sender's indirect utility is
$$
  v(x_1,x_2) = -(x_1-p)\cdot (x_2-p).
$$
Then $V(B) = 1/32$, and is achieved by
\begin{align*}
    \p = \frac{1}{8}\delta_{3/4,0} + \frac{3}{8}\delta_{3/4,1/2}+\frac{1}{8}\delta_{1/4,1}+\frac{3}{8}\delta_{1/4,1/2}.
\end{align*}
\end{proposition}
The distribution $\p$ is depicted in Figure~\ref{figure:anti}.

The proofs of Propositions~\ref{prop:quadratic-persuasion} and ~\ref{prop:anti} are presented in \S\ref{sec:proofs}. They do not rely on our characterization of feasible beliefs, and instead use the fact that conditional expectations are orthogonal projections in the Hilbert space of square-integrable random variables. This technique exploits the quadratic form of these persuasion problems. A natural avenue for future research is the extension of these techniques---or the development of new techniques---to tackle non-quadratic problems.

\begin{figure}[h]
\begin{center}
\begin{tikzpicture}[scale=0.4,line width=1pt]
        \draw (0,0)--(10,0)--(10,10)--(0,10)--(0,0);
        
        \draw[thin, densely dashed] (2.5,0)--(2.5,10);
		\draw[thin, densely dashed] (7.5,0)--(7.5,5);
		\draw[thin, densely dashed] (0,5)--(7.5,5);
		
		\draw [->,>=stealth] (0,0) -- (10.7,0);
\draw [->,>=stealth] (0,0) -- (0,10.7);
\node[below right] at (10.7,0) {\large $x_1$};
\node[above left] at (0,10.7) {\large $x_2$};
	    
	    \node[below] at (2.5,0) {\large $\frac{1}{4}$};
	    	    \node[below] at (0,0) {\large $0$};
   	    \node[below] at (10.1,0) {\large $1\phantom{\frac{1}{2}}$};
		\node[below] at (7.5,0) {\large $\frac{3}{4}$};
		\node[left] at (0,5) {\large $\frac{1}{2}$};
		\node[left] at (0,9.7) {\large $1$};

        \filldraw[blue] (7.5,0) circle (0.2);
        \node[above right] at (7.5,0) {\large $\frac{1}{8}$};
        \filldraw[red] (7.5,5) circle (0.4);
        \node[right] at (7.8,5) {\large $\frac{3}{8}$};
        \filldraw[red] (2.5,10) circle (0.2);
        \node[below left] at (2.5,10) {\large $\frac{1}{8}$};
        \filldraw[blue] (2.5,5) circle (0.4);
        \node[left] at (2.2,5) {\large $\frac{3}{8}$};

        \end{tikzpicture}
\end{center}
\caption{An optimal distribution of the first-order Bayesian persuasion problem given by the indirect utility $v(x_1,x_2) = -(x_1-p) \cdot (x_2-p)$ and prior $p=\half$; see Proposition~\ref{prop:anti}. For $p=\half$, this distribution achieves the lowest possible covariance between posterior beliefs: $-1/32$. Blue points describe the distribution of posteriors conditional on $\omega=0$. Red points describe the distribution of posteriors conditional on $\omega=1$.\label{figure:anti}}
\end{figure}
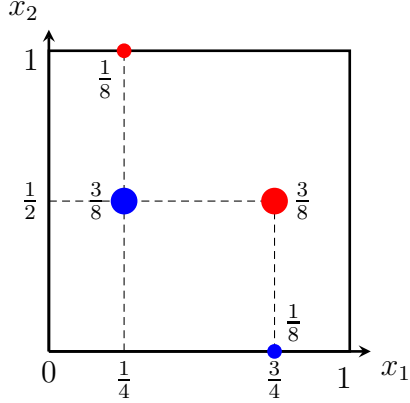

\paragraph{Extreme feasible joint posterior belief distributions.}
The set $\cP_p^N$ of $p$-feasible joint posterior belief distributions is a convex, compact subset of $\Delta([0,1]^N)$, when the latter is naturally equipped with the weak* topology; see  \cite{dubins1980maximal} or Proposition~\ref{prop:convex-compact} in the Appendix.  Together with Proposition~\ref{prop:social-persuasion}, this fact implies that $V(B)$ is always achieved at an {\em extreme point} of the set of $p$-feasible distributions $\cP_p^N$. It is thus natural to study the set of extreme points.

In the single agent case the concavification argument of \cite{aumann1995repeated} and \cite{kamenica2011bayesian} implies that every extreme point is a distribution with support of size at most $2$. This is not true for $2$ or more agents. For example, the posterior distribution with support of size $3$ defined in \eqref{eq:extreme-3} is extreme in $\cP_{1/2}^2$, since its restriction to any support of smaller cardinality is not feasible, as it cannot satisfy the martingale condition. The next theorem shows that there in fact exist extreme points with countably infinite support. It also states that the support cannot be too large, in the sense that every extreme point is supported on a set of zero Lebesgue measure.

\begin{theorem}
\label{thm:extreme}
Let $|N| \geq 2$. Then 
\begin{enumerate}
\item For every $p \in (0,1)$ there exists an extreme point in $\cP_p^N$ whose support has an infinite countable number of points. 
\item For every extreme $\p \in \cP_p^N$ there exists a measurable $A \subseteq [0,1]^N$ such that $\p(A)=1$, and the Lebesgue measure of $A$ is zero.
\end{enumerate}
\end{theorem}

To prove the first part of Theorem~\ref{thm:extreme} we explicitly construct an extreme feasible belief distribution with countably infinite support; see \S\ref{sec:proofs}. The construction is a variant of Rubinstein's e-mail game \citep{rubinstein1989electronic}. Unlike Rubinstein's email game, no agent in our construction is fully informed about the state. The resulting belief distribution for two agents is depicted in Figure \ref{figure:extreme}. Interestingly, it is possible to modify our construction in a way that places the beliefs closer and closer to the diagonal. This results in a sequence of extreme points that converge to a distribution that is supported on the diagonal, and has support size larger than two. Every extreme point that is supported on the diagonal is supported on at most two points, since this reduces to the single agent case. Therefore, for two agents or more, the set of extreme points is not closed within the set of feasible distributions. This demonstrates another distinction from the single agent case where the set of extreme points is closed.

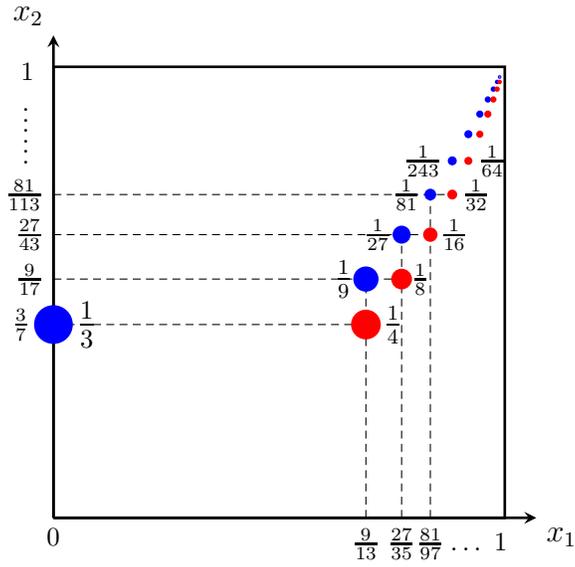
\begin{figure}
\begin{center}
\begin{tikzpicture}[scale=0.6, line width=1pt]
\begin{scope}[rotate=90,yscale=-1]
\draw (0,0)--(10,0)--(10,10)--(0,10)--(0,0);
\draw [->,>=stealth] (0,0) -- (10.7,0);
\draw [->,>=stealth] (0,0) -- (0,10.7);
\node[above left] at (10.7,0) {\large $x_2$};
\node[below right] at (0,10.7) {\large $x_1$};

\draw[thin, densely dashed] (4.28571,0)--(4.28571,6.92308);
\node[left] at (4.28571,0) {$\frac{3}{7}\phantom{a}$};
\draw[thin, densely dashed] (5.29412,0)--(5.29412,7.71429);
\node[left] at (5.29412,0) {$\frac{9}{17}$};

\draw[thin, densely dashed] (6.27907,0)--(6.27907,8.35052);
\node[left] at (6.27907,0) {$\frac{27}{43}$};
\draw[thin, densely dashed] (7.16814,0)--(7.16814,8.83636);
\node[left] at (7.16814,0) {$\frac{81}{113}$};
\node[left] at (8.3,0) {$\vdots\phantom{a}$};
\node[left] at (9,0) {$\vdots\phantom{a}$};
\node[left] at (9.9,0) {$1\phantom{i}$};

\draw[thin, densely dashed] (0,6.92308)--(5.29412,6.92308);
\node[below] at (0,6.92308) {$\frac{9}{13}$};
\draw[thin, densely dashed] (0,7.71429)--(6.27907,7.71429);
\node[below] at (0,7.71429) {$\frac{27}{35}$};
\draw[thin, densely dashed] (0,8.35052)--(7.16814,8.35052);
\node[below] at (0,8.35052) {$\frac{81}{97}$};
\node[below] at (0,0) {$0$};
\node[below] at (0,10.1) {$1\phantom{\frac{1}{2}}$};

\node[below] at (0,8.9) {$\phantom{\frac{2}{3}}\ldots$};

\filldraw[blue] (4.28571,0) circle (0.4);
\node[right] at (4.28571,0.3) {\Large $\frac{1}{3}$};
\filldraw[red] (4.28571,6.92308) circle (0.3);
\node[right] at (4.28571,7.1) {\large  $\frac{1}{4}$};
\filldraw[blue] (5.29412,6.92308) circle (0.25);
\node[left] at (5.29412,6.85) {\large $\frac{1}{9}$};
\filldraw[red] (5.29412,7.71429) circle (0.2);
\node[right] at (5.29412,7.71429) {$\frac{1}{8}$};
\filldraw[blue] (6.27907,7.71429) circle (0.17);
\node[left] at (6.27907,7.71429) {$\frac{1}{27}$};
\filldraw[red] (6.27907,8.35052) circle (0.13);
\node[right] at (6.27907,8.35052) {$\frac{1}{16}$};
\filldraw[blue] (7.16814,8.35052) circle (0.1);
\node[left] at (7.16814,8.35052) {$\frac{1}{81}$};
\filldraw[red] (7.16814,8.83636) circle (0.08);
\node[right] at (7.16814,8.83636) {$\frac{1}{32}$};
\filldraw[blue] (7.91531,8.83636) circle (0.07);
\node[left] at (7.91531,8.83636) {$\frac{1}{243}$};
\filldraw[red] (7.91531,9.19294) circle (0.06);
\node[right] at (7.91531,9.19294) {$\frac{1}{64}$};
\filldraw[blue] (8.50642,9.19294) circle (0.06);
\filldraw[red] (8.50642,9.44708) circle (0.05);
\filldraw[blue] (8.95211,9.44708) circle (0.05);
\filldraw[red] (8.95211,9.62447) circle (0.05);
\filldraw[blue] (9.27612,9.62447) circle (0.04);
\filldraw[red] (9.27612,9.74647) circle (0.04);
\filldraw[blue] (9.50548,9.74647) circle (0.03);
\filldraw[red] (9.50548,9.82954) circle (0.03);
\filldraw[blue] (9.6648,9.82954) circle (0.02);
\filldraw[red] (9.6648,9.88571) circle (0.02);
\filldraw[blue] (9.77401,9.88571) circle (0.01);
\end{scope}
\end{tikzpicture}

\end{center}
\caption{An extreme point of $\mathcal{P}^2_{1/2}$ with infinite countable support. The numbers near the points indicate their probabilities. Conditional on $\omega=\ell$ the pair of posteriors belongs to the set of blue points. Conditional on $\omega=h$ the pair of posteriors belongs to the set of red points.\label{figure:extreme}}
\end{figure}

The proof of the second part of Theorem~\ref{thm:extreme} relies on the classic result of \cite{lindenstrauss1965remark} regarding extreme points of the set of joint distributions with given marginals.

Theorem~\ref{thm:extreme} leaves a natural question open: Are there any non-atomic extreme points? Or conversely, does every extreme point have countable support?


\section{Implementing feasible distributions}
\label{sec:implementation}
Assume we are given a feasible distribution $\p\in\Delta([0,1]^n)$. A natural question is: Which information structures induce $\p$? And, relatedly, which conditional posterior belief distributions $(\p^\ell,\p^h)$ are compatible with $\p$?

By Lemma~\ref{lem:domination} (stated in the Appendix), $\p$ is feasible if and only if there exists a distribution $\q\in\Delta([0,1]^n)$ such that  $\q\leq \frac{1}{p}\p$ and $\dd\q_i(x)=\frac{x}{p}\dd \p_i(x)$ for all agents $i$. Note that this is a linear program, which in general is infinite-dimensional.

Given a solution $\q$, define $\p^h=\q$ and $\p^\ell=\frac{\p-p\cdot \q}{1-p}$. Then the distribution of posteriors $\p$ is induced by the information structure in which the signals have joint distribution $\p^\omega$ conditional on $\omega$ (see Lemma~\ref{lem:revelation} and its proof); in this case, we say that $\p$ is {\em implemented} by the pair $(\p^\ell,\p^h)$, which are also the conditional distributions of the beliefs. Thus, to find an information structure that induces $\p$, it suffices to solve the above mentioned linear program.

When $\p$ has finite support, this linear program has finite dimension, and thus a solution can be numerically calculated by simply applying a linear program solver. In the general, infinite dimensional case, we do not expect that simple, closed-form solutions always exist. An exception is the single agent case, in which an implementing pair is given by (see Lemma~\ref{lem:revelation})
$$
  \dd\p^\ell(x) = \frac{1-x}{1-p}\dd\p(x)\quad\quad\dd\p^h(x) = \frac{x}{p}\dd\p(x).
$$
We make two observations about the single agent case. First, there is a unique pair of conditional belief distributions $(\p^\ell,\p^h)$ that implements $\p$:  every information structure that induces $\p$ will have the same conditional distributions of beliefs. Second, the two distributions $\p^\ell,\p^h$ have the Radon–Nikodym derivative $\frac{\dd\p^h}{\dd\p^\ell}(x)=\frac{1-p}{p}\frac{x}{1-x}$, and so are mutually absolutely continuous, unless there is an atom on $0$ or on $1$. As we now discuss, neither of these two properties hold in general beyond the single agent case.

We consider the case that the number of agents is $n \geq 2$, and that $\p \in \Delta([0,1]^n)$ is a feasible distribution that admits a density. In this case, the next proposition shows that $\p$ can be implemented by $(\p^\ell,\p^h)$ that are very far from being mutually absolutely continuous: they are supported on disjoint sets.\footnote{These pairs are also easily seen to be extreme points among the convex set of pairs that implement $\p$.}
\begin{proposition}
\label{prop:gutman}
Let $n \geq 2$, and let $\p \in \Delta([0,1]^n)$ be $p$-feasible for some $p$. Assume that $\p$ admits a density. Then there exists a subset $D \subset [0,1]^n$, such that $\p$ can implemented by $(\p^\ell,\p^h)$ with the property that $\p^h$ is supported on $D$ and $\p^\ell$ is supported on the complement $\overline{D}$. Furthermore, restricted to $D$ and $\overline{D}$ respectively, $\p^\ell = \frac{1}{1-p}\p$, and $\p^h = \frac{1}{p}\p$. 
\end{proposition}
\begin{proof}
Let $R$ be a measure on $[0,1]^n$. For a subset $D \subseteq [0,1]^n$ we denote by $R\big|_D$ the restriction of $R$ to this subset, i.e., $R\big|_D(A)=R(D\cap A)$ for any measurable set $A$.

In Theorem 3 in \cite{gutmann1991existence} it was shown that for any absolutely continuous, finite measure $R$ on $[0,1]^n$ and every measure $\q\leq R$, there exists a measurable set $D$ such that $R\big|_D$ and $\q$ have identical marginals.

Apply this result to $R = \frac{1}{p}\p$ and $\q \leq \frac{1}{p}\p$, which solves the linear program of Lemma~\ref{lem:domination}. Then there is a $D \subset [0,1]^n$ such that $\q$ and $\frac{1}{p}\p\big|_D$ have the same marginals. Let $\q' = \frac{1}{p}\p\big|_D$. Then $\q'$ is also a solution to the linear program of Lemma~\ref{lem:domination}: $\q' \leq \frac{1}{p}\p$ and $\dd \q_i(x)=\frac{x}{p}\dd \p_i(x)$. Hence, $\p$ is implemented by $\p^h = \q' = \frac{1}{p}\p\big|_D$ and $\p^\ell = \frac{\p-p\cdot \q'}{1-p}=\frac{1}{1-p}\p|_{\overline D}$, which proves the claim.
\end{proof}

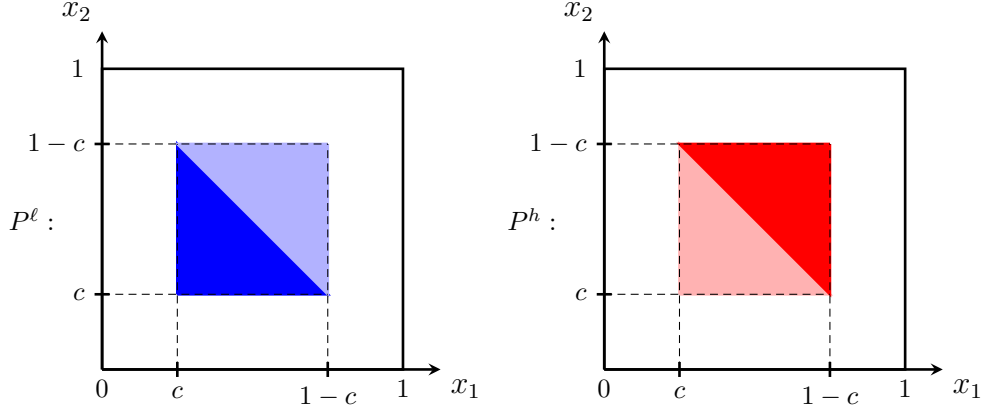
\begin{figure}
\begin{center}
\begin{tikzpicture}[line width=1pt]
\draw (0,0) -- (4,0)--(4,4)--(0,4)--(0,0);
		\draw [->,>=stealth] (0,0) -- (4.5,0);
\draw [->,>=stealth] (0,0) -- (0,4.5);
\node[below right] at (4.5,0) {\large $x_1$};
\node[above left] at (0,4.5) {\large $x_2$};

\filldraw[blue] (1,1)--(1,3)--(3,1)--(1,1);
\filldraw[blue!30!white] (3,3)--(1,3)--(3,1)--(3,3);
\node[below] at (0,0) {0};
\node[below] at (4,0) {1};

\draw (-0.1,1)--(0.1,1);
\node[left] at (-0.1,1) {$c$};
\draw (-0.1,3)--(0.1,3);
\node[left] at (-0.1,3) {$1-c$};
\node[left] at (-0.1,4) {$1$};
\draw[thin, densely dashed] (0,3)--(3,3);
\draw[thin, densely dashed] (0,1)--(3,1);
\draw[thin, densely dashed] (3,0)--(3,3);
\draw[thin, densely dashed] (1,0)--(1,3);

\draw (1,-0.1)--(1,0.1);
\node[below] at (1,-0.1) {$c$};
\draw (3,-0.1)--(3,0.1);
\node[below] at (3,-0.1) {$1-c$};

\node[left] at (-0.5,2) {$P^\ell :$};
\end{tikzpicture}
\begin{tikzpicture}[line width=1pt]
\draw (0,0) -- (4,0)--(4,4)--(0,4)--(0,0);
		\draw [->,>=stealth] (0,0) -- (4.5,0);
\draw [->,>=stealth] (0,0) -- (0,4.5);
\node[below right] at (4.5,0) {\large $x_1$};
\node[above left] at (0,4.5) {\large $x_2$};
\filldraw[red!30!white] (1,1)--(1,3)--(3,1)--(1,1);
\filldraw[red] (3,3)--(1,3)--(3,1)--(3,3);
\node[below] at (0,0) {0};
\node[below] at (4,0) {1};

\draw (-0.1,1)--(0.1,1);
\node[left] at (-0.1,1) {$c$};
\draw (-0.1,3)--(0.1,3);
\node[left] at (-0.1,3) {$1-c$};
\node[left] at (-0.1,4) {$1$};
\draw[thin, densely dashed] (0,3)--(3,3);
\draw[thin, densely dashed] (0,1)--(3,1);
\draw[thin, densely dashed] (3,0)--(3,3);
\draw[thin, densely dashed] (1,0)--(1,3);

\draw (1,-0.1)--(1,0.1);
\node[below] at (1,-0.1) {$c$};
\draw (3,-0.1)--(3,0.1);
\node[below] at (3,-0.1) {$1-c$};

\node[left] at (-0.5,2) {$P^h :$};
\end{tikzpicture}
\end{center}
\caption{The implementation $(P^\ell,P^h)$ of the uniform distribution on $[c,1-c]^2$. Darker/lighter colors indicate higher/lower densities. \label{figure:imp1}}
\end{figure}

When $n \geq 2$, an implementation is not always unique, in contrast to the single agent case. To see this, consider the case of $n =2$, and $\p$ equal to the uniform distribution $U_S$ on the square  $S=[c,1-c]^2$ with $c\in\left[0,\half\right)$. It is $\half$-feasible since the marginals are second-order dominated by the uniform distribution on $[0,1]$ (see Proposition~\ref{thm:product-feasible}). 
In \S\ref{sec:model}, we saw the implementation of the uniform distribution on $[0,1]^2$, which corresponds to $c=0$. For $c>0$, that implementation suggests that solutions of 
the linear program of Lemma~\ref{lem:domination} might be found in the form of a convex combination of $U_S$ and the uniform distribution $U_T$ on the upper-triangle $T=S\cap\{x_1+x_2\geq 1\}$. Indeed, it is easy to check that
$$\q=\frac{1-c}{1+c} \cdot U_T + \frac{2c}{1+c} \cdot U_S$$
has the correct marginals and satisfies $\q\leq 2\cdot  U_S$. The corresponding pair of distributions $(\p^\ell,\p^h)$ is  illustrated in Figure~\ref{figure:imp1}.

Note that in this implementation the two distributions $(\p^\ell,\p^h)$ do not have disjoint supports. Hence Proposition~\ref{prop:gutman} implies that another implementation exists. This is illustrated in Figure~\ref{figure:imp2}.

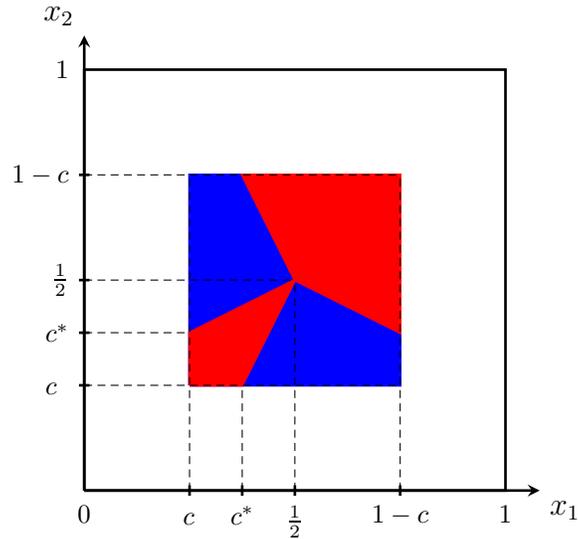
\begin{figure}[H]
\begin{center}
\begin{tikzpicture}[scale=0.7, line width=1pt]
\draw (0,0) -- (8,0)--(8,8)--(0,8)--(0,0);

\draw [->,>=stealth] (0,0) -- (8.65,0);
\draw [->,>=stealth] (0,0) -- (0,8.65);
\node[below right] at (8.65,0) {\large $x_1$};
\node[above left] at (0,8.65) {\large $x_2$};

\filldraw[blue] (2,2)--(2,6)--(6,6)--(6,2)--(2,2);
\filldraw[red] (2,2)--(2,3)--(4,4)--(3,2)--(2,2);
\filldraw[red] (4,4)--(6,3)--(6,6)--(3,6)--(4,4);
\draw (2,-0.1)--(2,0.1);
\node[below] at (2.1,-0.1) {$c\phantom{^*}$};
\draw[thin, densely dashed] (2,0)--(2,6);
\draw (3,-0.1)--(3,0.1);
\node[below] at (3,-0.1) {$c^*$};
\draw[thin, densely dashed] (3,0)--(3,2);
\draw (4,-0.1)--(4,0.1);
\node[below] at (4,-0.1) {$\frac{1}{2}$};
\draw[thin, densely dashed] (4,0)--(4,4);
\draw (6,-0.1)--(6,0.1);
\node[below] at (6,-0.1) {$1-c$};
\draw[thin, densely dashed] (6,0)--(6,6);
\node[below] at (0,-0.1) {$0$};
\node[below] at (8,-0.1) {$1$};

\draw (-0.1,2)--(0.1,2);
\node[left] at (-0.1,2) {$c\phantom{^*}$};
\draw[thin, densely dashed] (0,2)--(6,2);
\draw (-0.1,3)--(0.1,3);
\node[left] at (-0.1,3) {$c^*$};
\draw[thin, densely dashed] (0,3)--(2,3);
\draw (-0.1,4)--(0.1,4);
\node[left] at (-0.1,4) {$\frac{1}{2}$};
\draw[thin, densely dashed] (0,4)--(4,4);
\draw (-0.1,6)--(0.1,6);
\node[left] at (-0.1,6) {$1-c$};
\draw[thin, densely dashed] (0,6)--(6,6);
\node[left] at (-0.1,8) {$1$};
\end{tikzpicture}
\end{center}
\caption{An implementation of the uniform distribution on $[c,1-c]^2$ in which $\p^\ell$ (depicted in blue) and $\p^h$ (depicted in red) have disjoint supports. We denote $c^*=2c(1-c)$.  
\label{figure:imp2}
}
\end{figure}

\ifdefined\EC
\bibliographystyle{ACM-Reference-Format}
\fi
\bibliography{feasible_beliefs_bibliography}

\appendix
\section{Proof of Theorem~\ref{thm:two-agent-feasible}}
\label{sec:proof}

In this section we present a proof of Theorem~\ref{thm:two-agent-feasible}, a result due to \cite{dawid1995coherent}. We begin with the following simple lemma, which gives a direct revelation  principle for joint belief distributions. Similar statements appeared in the literature
\citep[e.g.,][]{gutmann1991existence, dawid1995coherent, levy2018persuasion, ziegler2020adversarial, arieli2020identifiable}. We include a proof for the convenience of the reader.
\begin{lemma}
\label{lem:revelation}
Let $\p \in \Delta([0,1]^n)$ be a $p$-feasible belief distribution. Then there exist $\p^\ell,\p^h \in \Delta([0,1]^n)$ such that $\p=(1-p)\cdot\p^\ell+p\cdot \p^h$ and for every $i \in \{1,\ldots,n\}$ the marginal distributions $\p^h_i$ and $\p_i$ satisfy
\begin{align}
    \label{eq:marginals}
    \dd\p^h_{i}(x)=\frac{x}{p}\dd\p_{i}(x).
\end{align}
In particular, every feasible belief distribution can be induced by an information structure in which, for each $i$, the belief $x_i$ is equal to the signal $s_i$.
\end{lemma}

The next lemma gives a necessary and sufficient condition for feasibility in terms of the existence of a measure with given marginals and with a given upper bound.
\begin{lemma}\label{lem:domination}
For $n\geq 2$ agents, a distribution $\p\in \Delta([0,1]^n)$ is $p$-feasible if and only if there exists a probability measure $\q \in \Delta([0,1]^n)$ such that
\begin{enumerate}
    \item $\q$ is upper-bounded by $\frac{1}{p}\p$, i.e.\ $\q(A) \leq \frac{1}{p}\p(A)$ for any measurable $A \subseteq [0,1]^n$, and
    \item for every $i \in \{1,\ldots,n\}$ the marginal distribution $\q_i$ is given by
    \begin{align}
        \label{eq:marginals2}
        \dd\q_{i}(x)=\frac{x}{p}\dd\p_{i}(x).
    \end{align}
\end{enumerate} 
\end{lemma}
\begin{proof}
If $\p$ is feasible, then by Lemma~\ref{lem:revelation} there is a pair $\p^\ell,\,\p^h \in \Delta([0,1]^n)$ such that $\p=(1-p)\p^\ell+p\cdot \p^h$ and such that \eqref{eq:marginals} holds. Picking $\q=\p^h$, it follows from \eqref{eq:marginals} that $\q\leq \frac{1}{p}\p$ and that $\q$ the desired marginals. In the opposite direction: if there is $\q$ satisfying the theorem hypothesis, then let $\p^h=\q$ and $\p^\ell=\frac{\p-p\cdot \q}{1-p}$. Consider the information structure $I$ in which $\omega=h$ with probability $p$, and $(s_1,\ldots,s_n)$ is chosen from $\p^h$ conditioned on $\omega=h$, and from $\p^\ell$ conditioned on $\omega=\ell$. It follows from \eqref{eq:marginals2} that $\p$ is the joint posterior distribution induced by $I$, and is hence $p$-feasible.
\end{proof}

Lemma~\ref{lem:domination} reduces the question of feasibility to the question of the existence of a bounded measure with given marginals. This is  a well-studied question in the case $n=2$. When the upper bound is proportional to the Lebesgue measure, \cite{lorentz1949problem} provided the answer in terms of first-order stochastic dominance of marginals. The discrete analog  of this result for matrices is known as the Gale-Ryser Theorem \citep{gale1957theorem,ryser1957combinatorial}. The condition for general upper-bound measures was derived by Kellerer \citeyearpar[Satz~4.2]{kellerer1961Funktionen}; the formulation below is due to Strassen  \citeyearpar[Theorem~6]{strassen1965existence}.
\begin{thmx}[Kellerer, Strassen]
Fix $\p \in \Delta([0,1]^2)$, $0 < p \leq 1$, and $M_1,\, M_2 \in \Delta([0,1])$. Then the following are equivalent:
\begin{enumerate}
    \item There exists a probability measure $\q \in \Delta([0,1]^2)$ that is bounded from above by $\frac{1}{p}\p$, and which has marginals $\q_i=M_i$.
    \item For every measurable $B_1,B_2 \subset [0,1]$ it holds that 
    \begin{equation}
        \label{eq_Strassen}
        \frac{1}{p}\p(B_1\times B_2)\geq M_1(B_1)+M_2(B_2)-1.
    \end{equation}
\end{enumerate}
\end{thmx}
As we now show, Theorem~\ref{thm:two-agent-feasible} is a consequence of Kellerer's Theorem and Lemma~\ref{lem:domination}.\footnote{We note that this result can also proved by interpreting the infinite-dimensional linear program for $\q$ from Lemma~\ref{lem:domination} as a condition that there is a feasible flow of magnitude~$1$ in an auxiliary network (as in \cite{gale1957theorem}, but with a continuum of edges) and then using max-flow/min-cut duality for infinite networks \citep[see][]{neumann1985theorem}.}

\begin{proof}[Proof of Theorem~\ref{thm:two-agent-feasible}]

When $\p$ is $p$-feasible, \eqref{eq:aumann-feasible-quant-2} follows immediately from Corollary~\ref{cor:agreement-quant}. To prove the other direction, we assume that $\p$ satisfies \eqref{eq:aumann-feasible-quant-2}, and prove that $\p$ is $p$-feasible for some $p$. To this end we will use Kellerer's Theorem to show that there exists a $\q \in \Delta([0,1]^2)$ satisfying the conditions of Lemma~\ref{lem:domination}.

Assume then that $\p$ satisfies \eqref{eq:aumann-feasible-quant-2}, i.e.\ for every measurable $A_1,A_2 \subseteq [0,1]$ it holds that
\begin{align}
\label{eq:feasibility_again}
\p(A_1 \times \overline{A_2}) \geq \int_{A_1} x\,\dd \p_1(x) - \int_{A_2} x\,\dd \p_2(x) 
  \geq \p(\overline{A_1} \times A_2).
\end{align}
Applying this to $A_1 = A_2 = [0,1]$ yields
$$
\int_0^1 x\,\dd \p_1(x) = \int_0^1 x\,\dd \p_2(x).
$$  
We accordingly set $p = \int_0^1 x\,\dd \p_1(x)$, and define the measures $M_1,\,M_2 \in \Delta([0,1])$ by $\dd M_i(x) = \frac{x}{p}\dd \p_i(x)$. Given any two measurable subsets $B_1,B_2 \subseteq [0,1]$, let $A_1 = B_1$ and $A_2 = \overline{B_2}$. Applying the left hand inequality of \eqref{eq:feasibility_again} to $A_1=B_1$ and $A_2 = \overline{B_2}$ yields
\begin{align*}
    \p(B_1 \times B_2) \geq \int_{B_1} x\,\dd \p_1(x) - \int_{\overline{B_2}} x\,\dd \p_2(x).
\end{align*}
Since $\dd M_i(x) = \frac{x}{p}\dd \p_i(x)$ we can write this as
\begin{align*}
\p(B_1 \times B_2) \geq p\cdot M_1(B_1) - p\cdot M_2(\overline{B_2}),
\end{align*}
dividing by $p$ and substituting $M_2(\overline{B_2})=1-M_2(B_2)$ yields.
\begin{align*}
  \frac{1}{p}\p(B_1 \times B_2) \geq M_1(B_1) + M_2(B_2)-1.
\end{align*}
Thus condition \eqref{eq_Strassen} holds, and we can directly apply Kellerer's theorem to conclude that there exists a measure $\q$ that satisfies the conditions of Lemma~\ref{lem:domination}, as condition \eqref{eq:marginals2} is simply $Q_i=M_i$. Hence $\p$ is $p$-feasible.
\end{proof}
\section{Proof of Theorem~\ref{thm:no-trade}}
\label{app:no-trade}

In this section we prove Theorem~\ref{thm:no-trade}. The proof of necessity is straightforward and explained before the theorem statement. The proof of sufficiency uses Th\'eor\`eme~2.6 of \cite{hansel1986probleme}, which is a generalization of Kellerer's Theorem.  

Recall that a {\em paving} of a set $X$ is a set of subsets of $X$ that includes the empty set, and an {\em algebra} is a paving that is closed under unions and complements. By a \emph{collection} we mean a multiset, i.e., the same element can enter the collection several times.
Denote $\bar\R_+ = [0,\infty]$. A finitely additive measure is a map from an algebra to $\bar\R_+$ that is additive for disjoint sets.
\begin{thmx}[\citet{hansel1986probleme}]
  \label{thm:hansel}
  Let $X$ be a set, and $\cF$ an algebra of subsets of $X$. Let $\cF_1,\ldots,\cF_n,\cG_1,\ldots,\cG_m$ be subpavings of $\cF$. For  $i \in \{1,\ldots,n\}$ and $j \in \{1,\ldots,m\}$ let  $\alpha_i \colon \cF_i \to \bar\R_+$,  and $\beta_j \colon \cG_j \to \bar\R_+$ be maps that vanish on the empty set. Then the following are equivalent:
  \begin{enumerate}
  \item There is a finitely additive measure $\q \colon \cF \to \R_+$ such that, for every $i \in \{1,\ldots,n\}$ and every $A \in \cF_i$ it holds that $\alpha_i(A) \leq \q(A)$, and for every $j \in \{1,\ldots,m\}$ and every $B \in \cG_j$ it holds that $\q(B) \leq \beta_j(B)$.
  \item For every finite collections of sets $\cA_1\subseteq\cF_1,\ldots,\cA_n\subseteq\cF_n$, $\cB_1\subseteq\cG_1,\ldots,\cB_m\subseteq\cG_m$, if
    \begin{align*}
      \sum_{i=1}^n\sum_{A \in \cA_i}\mathds{1}_{A} \leq \sum_{j=1}^m\sum_{B \in
      \cB_j}\mathds{1}_B
    \end{align*}
    then
    \begin{align*}
      \sum_{i=1}^n\sum_{A \in \cA_i}\alpha_i(A) \leq \sum_{j=1}^m\sum_{B \in
      \cB_j}\beta_j(B).
    \end{align*}
  \end{enumerate}
\end{thmx}

We will need the following corollary of this result.
\stepcounter{corx}
\begin{corx}
  \label{cor:hansel}
  Fix $\p \in \Delta([0,1]^n)$, $0 < p \leq 1$, and $M_1, \ldots, M_n \in \Delta([0,1])$. Then the following are equivalent:
  \begin{enumerate}
  \item There exists a probability measure $\q \in \Delta([0,1]^n)$ that is upper-bounded by $\frac{1}{p}\p$, and which has marginals $\q_i=M_i$.
  \item For every $\cA_1\ldots,\cA_n$, finite collections of Borel subsets of $[0,1]$, every $\cC$ a finite collection of Borel subsets of $[0,1]^n$, and for every non-negative integer $K$, if
    \begin{align}
      \label{eq:ht1a}
      \sum_{i=1}^n\sum_{A \in \cA_i}\mathds{1}_A(x_i) \leq K+ \sum_{C \in \cC}\mathds{1}_C(x_1,\ldots,x_n)
    \end{align}
    for all $(x_1,\ldots,x_n)$, then
    \begin{align}
      \label{eq:ht2a}
      \sum_{i=1}^n\sum_{A \in \cA_i}M_i(A) \leq K + \sum_{C \in \cC}\frac{1}{p}\p(C).
    \end{align}
  \end{enumerate}
\end{corx}
\begin{proof}
  The proof that 1.\ implies 2.\ is simple and omitted.\footnote{We will only make use of the other, non-trivial direction.}

  Let $\cF$ be the Borel sigma-algebra of $[0,1]^n$. For each $i \in \{1,\ldots,n\}$ let $\cF_i$ be the sub-sigma-algebra of sets that are measurable in the $i$\textsuperscript{th} coordinate. That is,  $\cF_i$ consists of sets of the form $[0,1]^{i-1} \times A \times [0,1]^{n-i}$, where $A$ is any Borel subset of $[0,1]$. Denote by $\pi_i \colon [0,1]^n \to [0,1]$ the projection on the $i$\textsuperscript{th} coordinate.

  Let $m=2$. Define $\cG_1$ as the trivial algebra $\{\emptyset,\ [0,1]^n\}$ and $\cG_2=\cF$. Let $\alpha_i = M_i$, $\beta_1\big([0,1]^n\big) = 1$, $\beta_1(\emptyset)=0$, and $\beta_2  = \frac{1}{p}\p$. Let $\cA_1\ldots,\cA_n$ and $\cC$ satisfy \eqref{eq:ht1a}.   By abuse of notation we identify each $A \in \cA_i$ with its preimage $\pi_i^{-1}(A) = [0,1]^{i-1} \times A \times [0,1]^{n-i}$. Thus $\cA_i \subseteq \cF_i$. We define $\cB_1$ and $\cB_2$ as follows. The collection $\cB_1\subseteq \cG_1$ contains $K$ copies of $[0,1]^n$ and $\cB_m = \cC\subseteq\cG_2$.

  We can apply Theorem~\ref{thm:hansel} directly to conclude that there is a finitely additive measure $\q$  that is upper-bounded by $\frac{1}{p}\p$, has $\q([0,1]^n)\leq 1$, and whose marginals $\q_i$ bound $M_i$ from above for  $i\in N$. Each  $M_i$ is a probability measure and the total mass of $\q$ is at most $1$; hence, $\q_i$ can upper-bound $M_i$ only if $\q_i=M_i$. We conclude that $\q$ is a probability measure with marginals $M_i$ for each $i$.
  Since $\q$ is upper-bounded by a sigma-additive measure, it is also itself sigma-additive (see Lemma~\ref{lem:sigma-additive} below). 
\end{proof}
\begin{lemma}
\label{lem:sigma-additive}
Let $\cF$ be a sigma-algebra of subsets of $X$. Let $\mu \colon \cF \to \bar\R_+$ be a finitely additive measure, and let $\nu \colon \cF \to \bar\R_+$ be a sigma-additive measure. If $\mu \leq \nu$ and $\nu(X) <\infty$ then $\mu$ is sigma-additive.
\end{lemma}
\begin{proof}
Let $A_1,A_2,\ldots \in \cF$ be pairwise disjoint. Denote $A = \cup_i A_i$. Then, by additivity, we have that
\begin{align*}
    \mu(A) = \mu(\cup_{i=1}^nA_i) + \mu(\cup_{i=n+1}^\infty A_i) \geq  \mu(\cup_{i=1}^nA_i) = \sum_{i=1}^n \mu(A_i)
\end{align*}
Hence $\mu(A) \geq \sum_{i=1}^\infty\mu(A_i)$. For the other direction, 
\begin{align*}
    \mu(A) = \mu(\cup_{i=1}^nA_i) + \mu(\cup_{i=n+1}^\infty A_i) \leq \sum_{i=1}^n\mu(A_i) + \nu(\cup_{i=n+1}^\infty A_i),
\end{align*}
since $\mu \leq \nu$. By the sigma-additivity and finiteness of $\nu$, the last addend vanishes as $n$ tends to infinity. Hence $\mu(A) \leq \sum_{i=1}^\infty\mu(A_i)$.
\end{proof}

\bigskip

We are now ready to finish the proof of Theorem~\ref{thm:no-trade}. We will show that \eqref{eq:trading} suffices for feasibility by checking that it implies the condition of Lemma~\ref{lem:domination}. Assume then that $\p$ satisfies \eqref{eq:trading},
i.e.\ for every trading scheme $(a_1,\ldots,a_i)$ it holds that
\begin{align}
  \label{eq:trading_again}
    \int_{[0,1]^n}\left(\sum_{i=1}^n a_i(x_i)x_i-\max\left\{0,\sum_{i=1}^n a_i(x_i)\right\}\right)\,\dd\p(x_1,\ldots,x_n) \leq 0.
\end{align}
Choose $i,j \in \{1,\ldots,n\}$, and consider the trading scheme in which $a_i=1$, $a_j=-1$ and $a_k=0$ for all $k \not \in \{i,j\}$. Then \eqref{eq:trading_again} implies that $\int_0^1 x\,\dd \p_i(x) = \int_0^1 x\,\dd \p_j(x)$. We accordingly set $p = \int_0^1 x\,\dd \p_i(x)$, and note that this definition is independent of the choice of $i$. Define the measures $M_i \in \Delta([0,1]^n)$ by $\dd M_i(x) = \frac{x}{p}\dd \p_i(x)$.

Let $\cA_1\ldots,\cA_n,$ $K$, and  $\cC$ satisfy \eqref{eq:ht1a}. We will show that \eqref{eq:ht2a} must hold. This will conclude the proof, since then it follows from Corollary~\ref{cor:hansel} that there exists a measure $\q$ that satisfies the conditions of Lemma~\ref{lem:domination}, as condition \eqref{eq:marginals2} is simply $Q_i=M_i$. 

To show that \eqref{eq:ht2a} holds, define the trading scheme $(a_1,\ldots,a_n)$ by
\begin{align*}
  a_i(x) = c\cdot\left(\sum_{A \in \cA_i}\mathds{1}_{A}(x)-\frac{K}{n}\right),
\end{align*}
where $c$ is a normalization constant chosen small enough so that the image of all the $a_i$'s is in $[-1,1]$. Then by \eqref{eq:trading_again}
\begin{align*}
  \int_{[0,1]^n}\sum_{i=1}^n a_i(x_i)x_i\,\dd\p(x_1,\ldots,x_n) \leq \int_{[0,1]^n}\max\left\{0,\sum_{i=1}^n a_i(x_i)\right\}\,\dd\p(x_1,\ldots,x_n).
\end{align*}
We substitute the definition of the trading scheme and the measures $M_i$. This yields
\begin{align*}
  p\cdot \sum_{i=1}^n\sum_{A \in \cA_i}M_i(A)-p\cdot K \leq \int_{[0,1]^n}\max\left\{0,\sum_{i=1}^n\sum_{A \in \cA_i}\mathds{1}_{A}(x_i)-K\right\}\,\dd\p(x_1,\ldots,x_n).
\end{align*}
By \eqref{eq:ht1a}, the integrand on the right-hand side is at most $\sum_{C \in \cC}\mathds{1}_C(x_1,\ldots,x_n)$. Hence
\begin{align*}
  p\cdot \sum_{i=1}^n\sum_{A \in \cA_i}M_i(A)-p\cdot K \leq \sum_{C \in \cC}\p(C).
\end{align*}
Dividing by $p$ and rearranging yields
\begin{align*}
  \sum_{i=1}^n\sum_{A \in \cA_i}M_i(A) \leq K +\sum_{C \in \cC}\frac{1}{p}\p(C).
\end{align*}
Thus condition \eqref{eq:ht2a} holds, and by Lemma~\ref{lem:domination} the distribution $\p$ is $p$-feasible. This completes the proof of Theorem~\ref{thm:no-trade}.

\bigskip
We end this section by noting a different avenue for proving this theorem. Lemma~\ref{lem:domination} gives a necessary and sufficient condition for a distribution $\p$ to be feasible, in terms of the existence of a distribution $\q$ satisfying two properties. Since the support of $\q$  is always a subset of the support of $\p$, for finitely supported $\p$ the conditions of the lemma reduce to a finite number of equalities and inequalities on a finite number of variables. Theorem~\ref{thm:no-trade} then becomes an corollary of the Farkas lemma.\footnote{A similar approach to the  characterization of finitely supported feasible distributions was taken by \cite{morris2020notrade}.} Theorem~\ref{thm:no-trade} for general distributions can then be deduced  by  approximation arguments.

\section{Additional missing claims and proofs}
\label{sec:proofs}

\begin{proof}[Proof of Lemma~\ref{lem:mps-monotone}]

To prove the lemma we assume that $\p'$ is not $p$-feasible and show that $\p$ is not $p$-feasible.

Let $(x_1,x_1',\ldots,x_n,x_n')$ be random variables such that (i) each pair $(x_i,x_i')$ is independent of the rest, (ii) $x_i$ has distribution $\mu_i$ and $x_i'$ has distribution $\mu_i'$, and (iii) $\mathbb{E}(x_i|x'_i)=x'_i$, which is possible since $\mu_i$ is a mean preserving spread of $\mu_i'$. 

By this definition, $\p$ is the joint distribution of $(x_1,\ldots,x_n)$ and $\p'$ is the joint distribution of $(x_1',\ldots,x_n')$. Hence, by Theorem~\ref{thm:no-trade}, in order to prove our claim it suffices to find a trading scheme $(a_1,\ldots,a_n)$ such that 
\begin{align}
  \label{eq:a_fails}
    \mathbb{E}\left(\sum_{i=1}^n a_i(x_i)x_i-\max\left\{0,\sum_{i=1}^n a_i(x_i)\right\}\right) > 0,
\end{align}
which would violate \eqref{eq:trading}.

Since $\p'$ is not $p$-feasible, it follows from Theorem~\ref{thm:no-trade} that there exists a trading scheme $(a_1',\ldots,a_n')$ such that
\begin{align}
  \label{eq:a_prime_fails}
    \mathbb{E}\left(\sum_{i=1}^n a_i'(x_i')x_i'-\max\left\{0,\sum_{i=1}^n a_i'(x_i')\right\}\right) > 0.
\end{align}
Define $a_i$ by\footnote{This only defines $a_i$ $\mu_i$-almost everywhere, which is sufficient for our needs; the rest of the definition can be completed using any version of the conditional expectation.}
$$
  a_i(x_i)=\mathbb{E}(a_i'(x'_i)|x_i).
$$
It follows from the definition of $a_i$, from the law of iterated expectation, and from  $\mathbb{E}(x_i|x'_i)=x'_i$ that
\begin{align}
  \mathbb{E}(a_i(x_i)\cdot x_i)
  &= \mathbb{E}\left(\mathbb{E}(a_i'(x'_i)|x_i)\cdot x_i\right) \nonumber \\
  &= \mathbb{E}\left(\mathbb{E}(a_i'(x'_i)\cdot x_i|x_i)\right) \nonumber \\
  &= \mathbb{E}\left(a_i'(x'_i)\cdot x_i\right) \nonumber \\
  &= \mathbb{E}\left(\mathbb{E}(a_i'(x'_i)\cdot x_i|x_i')\right) \nonumber \\
  &= \mathbb{E}\left(a_i'(x'_i)\cdot\mathbb{E}( x_i|x_i')\right) \nonumber \\
  &= \mathbb{E}\left(a_i'(x'_i)\cdot x_i'\right). \label{eq:transfer}
\end{align}
In addition, we have
\begin{align}
\mathbb{E}\left(-\max\left\{0,\sum_{i=1}^n a_i(x_i)\right\}\right)
&=\mathbb{E}\left(-\max\left\{0,\sum_{i=1}^n \mathbb{E}(a_i'(x'_i)|x_i)\right\}\right)\nonumber\\
&=\mathbb{E}\left(-\max\left\{0,\mathbb{E}\left(\sum_{i=1}^n a_i'(x'_i)\middle\vert x_1,\ldots,x_n\right)\right\}\right)\nonumber\\
&\geq\mathbb{E}\left(\mathbb{E}\left(-\max\left\{0,\sum_{i=1}^n a_i'(x'_i)\right\}\middle\vert x_1,\ldots,x_n\right)\right)\nonumber\\
&=\mathbb{E}\left(-\max\left\{0,\sum_{i=1}^n a_i'(x'_i)\right\}\right) \label{eq:loss}
\end{align}
where the first equality holds by the definition of $a_i$, the second equality holds since $(x_1,\ldots,x_n)$ are independent, the third inequality follows from Jensen's inequality since the function $x\mapsto -\max\{0,x\}$ is concave, and the last equality is again the law of iterated expectation.

Together, \eqref{eq:transfer} and \eqref{eq:loss} imply that the left hand side of \eqref{eq:a_fails} is at least as large as that of \eqref{eq:a_prime_fails}. Since the latter is positive, we conclude that
\begin{align*}
    \mathbb{E}\left(\sum_{i=1}^n a_i(x_i)x_i-\max\left\{0,\sum_{i=1}^n a_i(x_i)\right\}\right) \geq \mathbb{E}\left(\sum_{i=1}^n a_i'(x_i')x_i'-\max\left\{0,\sum_{i=1}^n a_i'(x_i')\right\}\right) > 0,
\end{align*}
as desired.
\end{proof}

\begin{proof}[Proof of Proposition~\ref{thm:product-feasible}]
Assume that the uniform distribution is a mean preserving spread of $\nu$ (or, equivalently, that $\nu$ is a {\em mean preserving contraction} of the uniform distribution). As shown before the statement of the proposition, the uniform distribution is $\half$-feasible. Hence $\p$ is $\half$-feasible by Lemma~\ref{lem:mps-monotone}.

Conversely, let $\nu$ be a probability measure that is symmetric around $\half$, and assume that $\nu$ is not a mean preserving contraction of the uniform distribution. We show  that $\nu\times\nu$ is not  $\half$-feasible. 

Let $F(x) = \nu([0,x])$ be the cumulative distribution function of $\nu$. Since the uniform distribution is not  a mean-preserving spread of $\nu$, there must be $y\in[0,1]$ such that
$$H(y)=\int_0^yF(x)\,\dd x-\int_0^y x\dd x =\int_0^yF(x)\,\dd x-\frac{y^2}{2}>0.
$$ 
Note that $H$ is continuous, and furthermore differentiable. Since $\nu$ is symmetric  and $\int_0^1F(x)\,\dd x=\half$, the function $H$ is also symmetric around $\half$: $H(y)=H(1-y)$. 

The  function $H$ vanishes at the end-points of the interval, and, as we noted above, is positive somewhere on $[0,1]$. Hence, it must have a global maximum $y \in (0,\half]$. By construction,
$H(y)>0$. Since $H$ is differentiable, we also have the first-order condition $H'(y)=F(y)-y=0$.

Let $z=\frac{1}{F(y)}\int_0^y x \,\dd \nu(x).$
We claim that $z<\frac{y}{2}$. To see this note that 
$$
  z=\frac{1}{F(y)}\int_0^y x\,\dd F(x)=\frac{1}{F(y)}\left[yF(y)-\int_0^y F(x) \,\dd x\right]<y-\frac{y^2}{2F(y)}=y-\frac{y}{2}=\frac{y}{2}.
$$
The second equality follows from integration by parts, the third equality holds since $\int_0^y F(x)\,\dd x>\frac{y^2}{2}$, and the forth inequality follows since $F(y)=y$. 
 
Let $A_1=[1-y,1]$ and $A_2=[0,y].$ Let $\p=\nu\times\nu$.
Note that by symmetry $\nu([y,1])=1-F(y)=1-y$ and by independence $\p(A_1\times\overline{A_2})=y(1-y)$. In addition, by construction, $\int_{A_1}x \,\dd \p(x)=\frac{1}{2}-\int_0^yx \,\dd \p(x)=\frac{1}{2}-yz$, and $\int_{A_2}x\,\dd \p(x)=y z$. Therefore 
$$
  \int_{A_1}x \,\dd \p(x)-\int_{A_2}x \,\dd \p(x)-\p(A_1 \times \overline{A_2})=\frac{1}{2}-yz -yz-y(1-z)=\frac{1}{2}-2yz-y+y^2>\frac{1}{2}-y\geq 0,
$$
where the second inequality follows since $z<\half$. Therefore, Theorem \ref{thm:two-agent-feasible} implies that  $\p=\nu\times\nu$ is not feasible.

\end{proof}

\begin{proof}[Proof of Proposition \ref{pro:product}] 
Let $m\in (0,1)$ be the median of $\nu$. We first prove the proposition for the case where $\nu$ has no atom at $m$. We assume for simplicity that the number of agents is even; i.e., $n=2k$ (otherwise, we simply ignore the last agent).

We consider the trading scheme $(a_1,\ldots,a_{2k})$ given by
\begin{align*}
    a_i(x) = \mathds{1}_{x < m}
\end{align*}
for $i \in \{1,\ldots, k\}$, and
\begin{align*}
    a_i(x) = -\mathds{1}_{x > m}
\end{align*}
for $ i \in \{k+1,\ldots, 2k\}$. We argue that for large enough $n$ (i.e., large enough $k$) this trading scheme violates the condition of Theorem~\ref{thm:no-trade}.

Let $B_1\sim \mathrm{Bin}(k,\,\frac{1}{2})$ denote the random variable of the number of agents $i=1,...,k$ whose posterior is in $[0,m)$ (i.e., those that sell the product). Let $B_2\sim \mathrm{Bin}(k,\,\frac{1}{2})$ denote the random variable of the number of agents $i=k+1,...,2k$ whose posterior is in $(m,1]$ (i.e., those that buy the product). Since $\nu^n$ is a product distribution we have that $B_1$ is independent of $B_2$. 

Note also that $\sum_{i=1}^n a_i(x_i) = B_2-B_1$, and so we have that the lower bound to the mediator's profit given by \eqref{eq:trading} equals
$$
  k\int_m^1x\,\dd\nu(x) - k\int_0^m x\,\dd\nu(x)-\mathbb{E}\big(\max\{0,B_2-B_1\}\big).
$$
We observe that $\mathbb{E}\big(\max\{0,B_2-B_1\}\big) \leq \sqrt{\frac{k}{8}}$: We have $\mathbb{E}\big(B_2-B_1\big)=0$ and $\mathbb{E}\big((B_2-B_1)^2\big)=\frac{k}{2}$. Jensen's inequality implies that $\mathbb{E}\big(|B_2-B_1|\big)\leq \sqrt{\frac{k}{2}}$ and by symmetry we deduce that $\mathbb{E}\big(\max\{B_2-B_1,0\}\big)\leq \frac{1}{2}\sqrt{\frac{k}{2}}$.

It follows that for $k> \frac{1}{8}\big(\int_m^1 x d\nu-\int_0^m x d\nu\big)^{-2}$ this lower bound on the mediator's profit is positive, and thus by Theorem~\ref{thm:no-trade} the belief distribution $\nu^n$ is not feasible.

In case the distribution $\nu$ has an atom at $m$, the only needed change is to choose $a_i(m)$ so that the expected number of units that each agent buys or sells is $\half$; the same calculation applies in this case.

\end{proof}

\begin{proof}[Proof of Lemma~\ref{lem:revelation}]
Let $I_0=((S_i)_{i \in N},\,\bP)$ be an information structure with prior $p$ that induces $\p$. By the law of total expectation,
$$
  x_i = \bP(\omega=h\mid s_i) = \bP(\omega=h\mid x_i)
$$
almost everywhere. Hence, if we define a new information structure $I'$ for which the signals are $s_i' = x_i$ and the beliefs are, accordingly, $x_i'=\bP(\omega=h\mid s_i')$, we will have that 
$$
  x_i' = s_i',
$$
and so we proved the second part of the claim. 

Denote by $\p^\ell$ and $\p^h$ be the conditional distributions of $(x_1',\ldots,x_n')$, conditioned on $\omega=\ell$ and $\omega=h$, respectively. Then clearly $\p = (1-p)\cdot \p^\ell+p \cdot \p^h$.

Finally,  \eqref{eq:marginals} holds, since by Bayes' law and the fact that $x_i' = s_i'$,
$$
  x_i' = \bP(\omega=h\mid x_i') = p\cdot \frac{\dd \p^h_i}{\dd \p_i}(x_i').
$$
\end{proof}

\begin{proof}[Proof of Proposition~\ref{prop:social-persuasion}]
  Let $B = (N,\,p,\,(A_i)_i,\,(u_i)_i,\, u_s)$ be a first-order Bayesian persuasion problem. By our equilibrium refinement assumption,  for each receiver $i$ there is a map $\alpha_i \colon [0,1] \to A_i$ such that in equilibrium receiver $i$ chooses action $\tilde a_i = \alpha_i(x_i)$ when their posterior is $x_i$. Hence, if we set
  $$
  v(x_1,\ldots,x_n) = u_s(\alpha_1(x_1),\ldots,\alpha_n(x_n)),
  $$
  Then, given an information structure chosen by the sender,
  \begin{align*}
    \E(u_s(\tilde a_1,\ldots,\tilde a_n)) = \E(u_s(\alpha_1(x_1),\ldots,\alpha_n(x_n))) = \E(v(x_1,\ldots,x_n)).
  \end{align*}
  Thus, if we denote by $\p$ the posterior belief distribution induced by the chosen information structure, we have shown that the sender's expected utility is
  \begin{align*}
    \int v(x_1,\ldots,x_n)\,\dd \p(x_1,\ldots,x_n).
  \end{align*}
  Since the sender, by choosing the information structure, can choose any $p$-feasible distribution, we arrive at the desired conclusion.
\end{proof}

\begin{proof}[Proof of Proposition~\ref{prop:quadratic-persuasion}] 
We first show that every information structure $I = (S_1,S_2,\bP)$ with prior $p$ chosen by the sender yields utility at most $(1-p)p$. 

This proof uses the fact conditional expectations are orthogonal projections in the Hilbert space of square-integrable random variables. We now review this fact.

For the probability space $(\Omega \times S_1 \times S_2,\, \bP)$ denote by $\cL^2$ the Hilbert space of square-integrable random variables, equipped with the usual inner product $(X,\,Y) = \E(X\cdot Y)$ and corresponding norm $\Vert X \Vert = \sqrt{\E(X^2)}$.

Given a sub-sigma-algebra $\mathcal{G} \subseteq \mathcal{F}$, denote by $\cL^2(\mathcal{G}) \subseteq \cL^2$ the closed subspace of $\mathcal{G}$-measurable, sqaure-integrable random variables. Recall that the map $X \mapsto \E(X \mid\mathcal{G})$ is simply the orthogonal projection $\cL^2 \to \cL^2(\mathcal{G})$. 

The following is an elementary lemma regarding Hilbert spaces. We will, of course, apply it to $\cL^2$.
\begin{lemma}
  \label{lem:ortho}
  Let $u$ be a vector in a Hilbert space $U$, and let $w$ be an orthogonal projection of $u$ to a closed subspace $W \subseteq U$. Then $w$ lies on a sphere of radius $\frac{1}{2}\Vert u\Vert$ around $\frac{1}{2}u$, i.e.\ $\Vert w -  \frac{1}{2}u\Vert = \frac{1}{2}\Vert u\Vert$.
\end{lemma}
\begin{proof}
  Since $w$ is an orthogonal projection of $u$, we can write $u = w + w'$, where $w$ and $w'$ are orthogonal. This orthogonality implies that $
    \Vert w-w' \Vert = \Vert w+w'\Vert = \Vert u \Vert$. It follows that
  \begin{align*}
    \Big\Vert w -  \frac{1}{2}u\Big\Vert= \Big\Vert w -  \frac{1}{2}(w+w')\Big\Vert= \Big\Vert \frac{1}{2}w -  \frac{1}{2}w'\Big\Vert= \frac{1}{2}\Vert w -  w'\Vert= \frac{1}{2}\Vert u\Vert.
  \end{align*}
\end{proof}

Let $X = \mathds{1}_{\omega = h}$ be the indicator of the event that the state is high. Hence $x_i = \E(X\mid\mathcal{F}_i)$, where $\mathcal{F}_i = \sigma(s_i)$ is the sigma-algebra generated by the private signal $s_i$. Denote $\hat X = X/p - 1$. Hence $\bP(\hat X=1/p-1)=p$, $\bP(\hat X=-1)=1-p$, and $\Vert \hat X \Vert^2 = (1-p)/p$. Denote $\hat X_i = x_i/p-1$, so that $\hat X_i = \E(\hat X\mid\mathcal{F}_i)$. Since $\hat X_i$ is the projection of $\hat X$ to the subspace $\cL^2(\mathcal{F}_i)$, it follows from Lemma~\ref{lem:ortho} that
\begin{align*}
  \Vert \hat X_i - \hat X/2 \Vert = \frac{1}{2}\sqrt{\frac{1-p}{p}}.
\end{align*}
Therefore, by the triangle inequality, $\Vert \hat X_1 - \hat X_2 \Vert \leq \sqrt{(1-p)/p}$, and since $\hat X_i = x_i/p-1$ we get that $\Vert x_1/p - x_2/p \Vert \leq \sqrt{(1-p)/p}$,
or
\begin{align*}
  \E((x_1/p-x_2/p)^2) \leq \frac{1-p}{p}.
\end{align*}
Thus
\begin{align*}
  \E((x_1-x_2)^2) \leq (1-p)p.
\end{align*}

To finish the proof of  Proposition~\ref{prop:quadratic-persuasion}, we note that a simple calculation shows that the sender can achieve the expected utility $(1-p)p$ by completely informing one agent, and giving no information to the other. That is, by inducing the joint posterior distribution
$$
\p = p\cdot\delta_{1,p} + (1-p)\delta_{0,p}.
$$
\end{proof}

\paragraph{Remark on the polarizing first-order Bayesian persuasion problem with $a < 2$ and $p=\half$.}
For the first-order Bayesian persuasion problem $B_a$ given by $v(x_1,x_2)=|x_1-x_2|^a$ and symmetric prior $p=\half$, we argue that $\p=\frac{1}{2}\delta_{1,\frac{1}{2}} + \frac{1}{2}\delta_{0,\frac{1}{2}}$ is the optimal posterior distribution for all $a\in (0, 2)$; thus the value $V(B_a)$ of the problem  equals $\frac{1}{2^a}$. To prove this, it is enough to check the upper bound $V(B_a)\leq \frac{1}{2^a}$.

Consider an arbitrary $\half$-feasible policy $P'$. H\"{o}lder's inequality implies
$$\int v(x_1,x_2)\,d\p'= \int v(x_1,x_2)\cdot 1\,d\p'\leq \left(\int \left(v(x_1,x_2)\right)^q\,d\p' \right)^\frac{1}{q}\cdot\left(\int 1^{q'}\, d\p'\right)^\frac{1}{q'}$$
where $q,q'>0$ and $\frac{1}{q}+\frac{1}{q'}=1$.  Picking $q=\frac{2}{a}$ and taking supremum over $\p'$ on both sides, we get $V(B_a)\leq \big(V(B_2)\big)^\frac{a}{2}$. By Proposition~\ref{prop:quadratic-persuasion}, $V(B_2)=\frac{1}{4}$ and we obtain the required upper bound.

\begin{proof}[Proof of Proposition~\ref{prop:anti}]
It is immediate to check that the expected utility for the distribution $\p$ given in the theorem statement is $1/32$. It thus remains to prove that no other distribution can achieve higher expected utility.

We retain the notation used in the proof of Proposition~\ref{prop:social-persuasion}. Let $(x_1,x_2)$ be the posterior beliefs of two agents induced by some information structure. Since the prior is $p=\half$, $\hat X_i = 2x_i-1$, and so
\begin{align*}
    \mathbb{E}(v(x_1,x_2)) = -\mathbb{E}((x_1-1/2) \cdot (x_2-1/2)) = -\frac{1}{4}\mathbb{E}(\hat X_1 \cdot \hat X_2).
\end{align*}
We recall that $(\cdot,\cdot)$ denotes the inner product in the Hilbert space $\cL^2$. Thus we can write
\begin{align}
    \label{eq:x1x2}
    \mathbb{E}(v(x_1,x_2)) = -\frac{1}{4}(\hat X_1,\, \hat X_2).
\end{align}
By Lemma~\ref{lem:ortho}, the vectors $\hat X_1,\hat X_2 \in \cL^2$ lie on a sphere of radius $\frac{1}{2}\Vert \hat X\Vert$ around $\frac{1}{2}\hat X \in \cL^2$. Note that since $p=\half$, $\Vert\hat X\Vert=1$. Since $\{\hat X,\hat X_1,\hat X_2\}$ span a subspace that is (at most) three dimensional, the question is reduced to the following elementary question in  three dimensional geometry: given a vector $w \in \R^3$, what is the minimum of the inner product $(u_1,u_2)$, as $u_1,u_2 \in \R^3$ range over the sphere of radius $\Vert w \Vert$ around $w$? Lemma~\ref{lem:sphere} states that the minimum is $-\frac{1}{2}\Vert w \Vert^2$. Hence by \eqref{eq:x1x2}
\begin{align*}
    \mathbb{E}(v(x_1,x_2)) \leq \frac{1}{32}.
\end{align*}
\end{proof}

\begin{figure}[h]
\begin{center}
\begin{tikzpicture}[scale=0.4, line width = 0.7pt]
\draw  (0,2)--(12,2);
\draw  (2,0)--(2,12);

\draw (5.54, 5.54) circle (5);

\draw[line width = 1pt,->, >=stealth] (2,2) -- (5.54,5.54);
\node[right] at (5.54,5.54) {\large $w$};


\node[below left] at (2,2) {\large $0$};
\draw[line width = 1pt,->, >=stealth] (2,2) -- (0.71,6.83);
\node[right] at (0.69,6.87) {\large $u_1$};
\draw[line width = 1pt,->, >=stealth] (2,2) -- (6.83,0.71);
\node[above] at (6.83,0.71) {\large $u_2$};
\end{tikzpicture}
\end{center}
\caption{Lemma~\ref{lem:sphere} states that $-\half$ is the smallest possible inner product between two vectors that lie on a unit sphere which intersects the origin. The depicted $u_1$ and $u_2$ achieve this minimum.\label{figure:sphere}}
\end{figure}
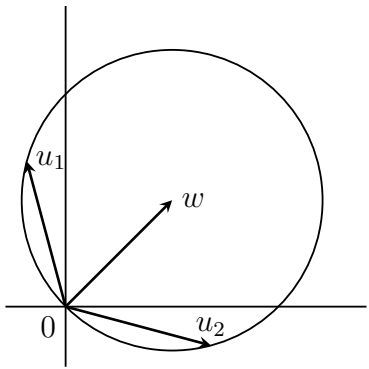

\begin{lemma}
\label{lem:sphere}
Let $w \in \R^3$. Then
$$
  \min\Big\{(u_1,u_2)\,:\, \Vert w-u_1\Vert = \Vert w-u_2\Vert = \Vert w \Vert \Big\} = -\frac{1}{2}\Vert w \Vert^2.
$$  
\end{lemma}
The proof is left as an exercise to the reader. Alternatively, this problem can be solved symbolically using the Mathematica command \begin{verbatim}
  Minimize[{x1*y1 + x2*y2 + x3*y3, 
  (x1 - 1)^2 + (x2)^2 + (x3)^2 == 1 && (y1 - 1)^2 + (y2)^2 + (y3)^2 == 1}, 
  {x1, x2, x3, y1, y2, y3}]
\end{verbatim}

\begin{proof}[Proof of Theorem \ref{thm:extreme}]
We start by constructing an extreme point of $\cP_p^2$ with infinite support. Thereafter, we extend this construction for $N\geq 2$ receivers. 

Consider the following information structure. Let $S_1 = S_2 = \{1,2,\ldots\}$. Choose at random a number $K$, which is distributed geometrically with parameter $r$ that depends on the state: When $\omega=\ell$, $r=\nicefrac{2}{3}$, and when $\omega=h$, $r=\nicefrac{1}{2}$. The signals $(s_1,s_2)$ that the agents receive are equal to $(K,K)$ when $\omega=\ell$, and to $(K+1,K)$ when $\omega=h$. 

A simple calculation shows that when an agent 1 observes a signal $s_1=k$ her posterior $x_1(k)=t_k$
and when agent 2 observes $s_2=k$, her posterior is $x_2(k)=w_k$, where $t_k$ and $w_k$ are given by 
$$t_k=\left\{\begin{array}{cc}\frac{2^{-k+1}p}{2^{-k+1}p+2\cdot 3^{-k}(1-p)}, & k\geq 2\\
0, & k=1\end{array}\right.,\qquad w_k=\frac{2^{-k}p}{2^{-k}p+2\cdot 3^{-k}(1-p)},\ \   k\geq 1.$$
The induced conditional distributions of $(x_1,x_2)$ are
\begin{align*}
    \p^{\ell} &= \sum_{k=1}^\infty 2\cdot 3^{-k} \cdot \delta_{(t_k,\,w_k)} \\
    P^{h} &= \sum_{k=1}^{\infty} 2^{-k} \cdot \delta_{(t_{k+1},\,w_k)}.
\end{align*}
Note that $\p^\ell$ and $\p^h$ have disjoint supports; see Figure~\ref{figure:extreme}.

Hence the induced (unconditional) distribution is $\p=(1-p)\cdot \p^\ell+p\cdot \p^h$. It satisfies the  pair of identities
\begin{align}\label{eq_p_extreme_identity_t}
{t_{k}}\cdot \p\left(\left\{\big(t_{k},w_{k}\big)
\right\}\right)
&=(1-t_{k})\cdot \p\left(\left\{\big(t_{k},w_{k-1}\big)\right\}\right), & k\geq 2,\\
\label{eq_p_extreme_identity_w}
{w_{k}}\cdot \p\left(\left\{\big(t_{k},w_{k}\big)
\right\}\right)
&=(1-w_{k})\cdot \p\left(\left\{\big(t_{k+1},w_{k}\big)\right\}\right),& k\geq 1,
\end{align}
which we use below.

We first argue that $(\p^\ell,\,\p^h)$ is the unique implementation of $P$.\footnote{See \S\ref{sec:implementation} for the definition of an implementation, as well as an example that shows that feasible distributions may have multiple implementations.} 

By Lemma~\ref{lem:revelation}, for any pair $(\q^\ell,\q^h)$ implementing a feasible distribution $\q$, the marginals 
satisfy $\dd \q^h_i(x)=\frac{x}{p}\dd \q_i(x)$ for each agent $i\in N$. Since $\q^\ell=\frac{\q-p\cdot \q^h}{1-p}$, we get the complementary equation $\dd \q^\ell_i(x)=\frac{1-x}{1-p}\dd \q_i(x)$. Combining the two equations, we obtain 
\begin{equation}\label{eq_identity_for_marginals}
\frac{x}{p}\cdot \dd \q^\ell_i(x)=\frac{1-x}{1-p}\cdot \dd \q^h_i(x) \ \ \ \mbox{ for $i\in N$.}
\end{equation}

Let $(\hat\p^\ell,\hat\p^h)$ be an implementation of $\p$. The identity~\eqref{eq_identity_for_marginals}, applied to $i=1$ and $x=t_k$, gives the following equation
\begin{equation}\label{eq_hatP_base}
\frac{t_{1}}{p} \hat\p^\ell\left(\left\{\big(t_{1},w_{1}\big)
\right\}\right)=\frac{1-t_{1}}{1-p}\hat\p^h\left(\left\{\big(t_{1},w_{1}\big)
\right\}\right)
\end{equation}
for $k=1$, and the following family of equations for $k\geq 2$:
\begin{align}\label{eq_projection_x1}
\lefteqn{\frac{t_{k}}{p} \left(\hat\p^\ell\left(\left\{\big(t_{k},w_{k}\big)
\right\}\right) + \hat\p^\ell\left(\left\{\big(t_{k},w_{k-1}\big)
\right\}\right)\right)}\notag\\
&=\frac{1-t_{k}}{1-p}\left(\hat\p^h\left(\left\{\big(t_{k},w_{k}\big)
\right\}\right) + \hat\p^h\left(\left\{\big(t_{k},w_{k-1}\big)
\right\}\right)\right).
\end{align}
Similarly, for $i=2$ and $x_2=w_k$, we get
\begin{align}
\label{eq_projection_x2}
\lefteqn{\frac{w_{k}}{p} \left(\hat\p^\ell\left(\left\{\big(t_{k},w_{k}\big)
\right\}\right) + \hat\p^\ell\left(\left\{\big(t_{k+1},w_{k}\big)
\right\}\right)\right)}\notag\\
&=\frac{1-w_{k}}{1-p}\left(\hat\p^h\left(\left\{\big(t_{k},w_{k}\big)
\right\}\right) + \hat\p^h\left(\left\{\big(t_{k+1},w_{k}\big)
\right\}\right)\right).
\end{align}
We now show that these equations and the condition $\p=(1-p)\hat\p^\ell+ p\cdot\hat\p^h$ completely determine the pair $(\hat\p^\ell, \hat\p^h)$. Since $t_1=0$, equation~\eqref{eq_hatP_base} results in $\hat\p^h(\{(t_1,w_1)\})=0$ and hence the entire mass
$\frac{1}{1-p}\p(\{(t_1,w_1)\})$ of the point $(t_1,w_1)$ must be assigned to $\hat{P}^\ell$. Given this, the equality~\eqref{eq_projection_x2} implies the entire mass $\frac{1}{p}\p(\{(t_{2},w_1)\})$ of the point $(t_2,w_1)$ must be assigned to $\hat{\p}^h$. Indeed, expressing $\hat\p^\ell(\{(t_k,w_k)\})$ and $\hat\p^h(\{(t_{k+1},w_k)\})$ through $\p$ and taking into account that $\hat\p^h(\{(t_k,w_k)\})=0$ for $k=1$, we rewrite~\eqref{eq_projection_x2} as
\begin{align}\label{eq_inductive_w}
\lefteqn{\frac{w_{k}}{p} \left(\frac{1}{1-p}\p\left(\left\{\big(t_{k},w_{k}\big)
\right\}\right) + \hat\p^\ell\left(\left\{\big(t_{k+1},w_{k}\big)
\right\}\right)\right)}\notag\\
&=\frac{1-w_{k}}{1-p}\left(\frac{\p\left(\left\{\big(t_{k+1},w_{k}\big)
\right\}\right)-(1-p)\hat\p^\ell\left(\left\{\big(t_{k+1},w_{k}\big)
\right\}\right)}{p}\right)
\end{align}
for $k=1$. This equality 
and the identity~\eqref{eq_p_extreme_identity_w} lead to $\hat\p^\ell(\{(t_{2},w_1)\})=0$.

Next, a similar argument demonstrates that the entire mass $\frac{1}{1-p}\p\left(\left\{\big(t_{2},w_{2}\big)\right\}\right)$ of the point $(t_2,\,w_2)$ must be assigned to $\hat{\p}^\ell$.  Indeed, we know that $\hat\p^\ell(\{(t_k,w_{k-1})\})=0$ for $k=2$, which allows us to rewrite~\eqref{eq_projection_x1} as
\begin{align}\label{eq_inductive_t}
\lefteqn{\frac{t_{k}}{p} \left(\frac{\p\left(\left\{\big(t_{k},w_{k}\big)
\right\}-p\cdot\hat\p^h\left(\left\{\big(t_{k},w_{k}\big)
\right\}\right)\right)}{1-p}\right)}\notag\\
&=\frac{1-t_{k}}{1-p}\left(\hat\p^h\left(\left\{\big(t_{k},w_{k}\big)
\right\}\right) + \frac{1}{p}\p\left(\left\{\big(t_{k},w_{k-1}\big)
\right\}\right)\right)
\end{align}
for $k=2$. Combining this equality with  identity~\eqref{eq_p_extreme_identity_t}, we get $\hat\p^h\left(\left\{\big(t_{2},w_{2}\big)
\right\}\right)=0$.
We proceed inductively: knowing that $\hat\p^h\left(\left\{\big(t_{k},w_{k}\big)
\right\}\right)=0$, we deduce equality~\eqref{eq_inductive_w} and infer using~\eqref{eq_p_extreme_identity_w} that $\hat\p^\ell\left(\left\{\big(t_{k+1},w_{k}\big)
\right\}\right)=0$; from this, we derive equality~\eqref{eq_inductive_t} and, with the help of~\eqref{eq_p_extreme_identity_t}, see that $\hat\p^h\left(\left\{\big(t_{k+1},w_{k+1}\big)
\right\}\right)=0$, and so on. Thus $\hat\p^\ell$ coincides with $\frac{1}{1-p}\p $ restricted to  $\{(t_k,w_k)\,:\,k\geq 1\}$ and $\hat\p^h$ coincides with $\frac{1}{p}\p$ restricted to $\{(t_{k+1},w_k)\,:\, k \geq 1\}$, and so $\hat \p^\ell = \p^\ell$, $\hat\p^h=\p^h$. Hence $\p$ has a unique implementation.

Given  a convex combination $\p=\alpha \q + (1-\alpha) \hat{\q}$, where $\q$ and $\hat{\q}$ are both feasible and $\alpha\in(0,1)$, our goal is to show that $\q=\p$.
Let $(\q^\ell,\,\q^h)$ and $(\hat{\q}^\ell,\,\hat{\q}^h)$
be some pairs of conditional probability distributions that implement $\q$ and $\hat{\q}$, respectively.
The pair $(\alpha \q^\ell+ (1-\alpha) \hat{\q}^\ell,\,\alpha \q^h+ (1-\alpha) \hat{\q}^h)$ implements~$\p$. From the uniqueness of $\p$'s implementation we deduce that $\supp(\q^\ell)\subset \supp(\p^\ell)$ and $\supp(\q^h)\subset \supp(\p^h)$, and thus $\q^\ell(\{(t_{k+1},w_k)\})=\q^h(\{(t_{k},w_k)\})=0$. Therefore, condition~\eqref{eq_identity_for_marginals} implies 
\begin{align}
\notag
\frac{t_{k}}{p}\cdot  \q^\ell\left(\left\{\big(t_{k},w_{k}\big)
\right\}\right)&=
\frac{1-t_{k}}{1-p}\cdot  \q^h\left(\left\{\big(t_{k},w_{k-1}\big)\right\}\right), & k&\geq 2\\
\notag
\frac{w_{k}}{p}\cdot  \q^\ell\left(\left\{\big(t_{k},w_{k}\big)
\right\}\right)&=
\frac{1-w_{k}}{1-p}\cdot\q^h\left(\left\{\big(t_{k+1},w_{k}\big)
\right\}\right), & k&\geq 1.
\end{align}
This family of equations uniquely determines the weights of each point
$(t_k,w_k)$ and $(t_{k+1},w_k)$, $k\geq 1$, up to a multiplicative factor, which is pinned down by the condition that $\q$ is a probability measure. Hence in the decomposition $\p=\alpha \q + (1-\alpha) \hat{\q}$, the distribution $\q$ is unique. Thus $\q$ must be equal to $\p$ and hence $\p$ is an extreme point.

This construction extends to $N\geq 2$, by sending no information to the remaining $N-2$ agents. The resulting distribution over posteriors is an extreme point due to the same arguments.

Now we check the singularity of the extreme points of $ \cP_p^N$ with respect to the Lebesgue measure $\lambda$ on the unit cube $[0,1]^N$. The proof relies on the classical theorem of \citet{lindenstrauss1965remark}, which states that all extreme points of the set of all probability measures with given marginals are singular. \citet{lindenstrauss1965remark} proved this result for $n=2$ and \citet{shortt1986singularity} extended it to general $n$. In~\S\ref{sec_Linden} (Theorem~\ref{th_Linden}), we include  an alternative proof for $n\geq 2$ closely following the original proof of Lindenstrauss. Here we show how to apply this result  to our problem.

By Lemma~\ref{lem:revelation}, any $p$-feasible distribution $\p$ can be represented as $(1-p)\p^\ell+p\cdot \p^h$, where the marginals of the pair $(\p^\ell,\, \p^h)$ satisfy the identity
 \begin{equation}\label{eq_pair_feasibility}
  \frac{x}{p}dP_i^\ell(x)=\left(1-\frac{x}{p}\right)dP_i^h(x).
 \end{equation} 
 Denote by ${\mathcal{H}}_p^N$ the set of all pairs satisfying~\eqref{eq_pair_feasibility}. The set of $p$-feasible distributions $\cP_p^N$ is, therefore, the image of the convex set $\mathcal{H}_p^N$ under the linear map $(\p^{\ell},\,\p^h)\to (1-p)\p^\ell+p\cdot \p^h$. The  extreme points of ${\cP}_p^N$ are hence all contained in the image of the extreme points of ${\mathcal{H}}_p^N$.
 
So it is enough to check that in any extreme pair $(\p^\ell,\,\p^h)$, both measures are necessary singular with respect to $\lambda$. We check singularity of $\p^\ell$; the argument for $\p^h$ is analogous. We assume towards a contradiction that $\p^\ell$ is not singular. Then by Theorem~\ref{th_Linden}, $\p^\ell$ is not an extreme point among measures with the same marginals and thus can be represented as $\p^\ell=\frac{1}{2}Q+\frac{1}{2}Q'$, where $Q$ and $Q'$ have the same marginals as $\p^\ell$ and $Q\ne Q'$. This induces the representation of $(\p^\ell,\,\p^h)$ as the average of $\big(Q, \,\p^h\big)$ and $\big(Q', \,\p^h\big)$, where both pairs satisfy~\eqref{eq_pair_feasibility}. Contradiction.
\end{proof}
\begin{proposition}
\label{prop:convex-compact}
The set of $p$-feasible posterior distributions is a convex, weak* compact subset of $\Delta([0,1]^n)$.
\end{proposition}
\begin{proof}
Pick a pair of feasible distributions $\p,\,\p' \in \cP_p^N$ and consider the corresponding $Q,\,Q'\in\Delta\big([0,1]^N\big)$ that satisfy the conditions of Lemma~\ref{lem:domination} for $\p$ and $\p'$, respectively. For a convex combination $\p''=\alpha\p+(1-\alpha)\p'$, the conditions  of the lemma  are satisfied by $Q''=\alpha Q+(1-\alpha)Q'$. Thus $\p''$ also belongs to $\cP_p^N$, and so we have shown that the set of feasible distributions is convex.

To verify the weak* compactness of $\cP_p^N$, consider a sequence of $p$-feasible distributions $\p^{(k)}$ weakly converging to some $\p^{(\infty)}\in \Delta\big([0,1]^N\big)$, as $k\to\infty$. To prove weak* compactness we show that the limit distribution $\p^{(\infty)}$ is also feasible. For each $\p^{(k)}$, select some $Q^{(k)}$ from Lemma~\ref{lem:domination}. The set of all probability distributions on $[0,1]^N$ is weak* compact and, therefore, there is a subsequence $Q^{(k_m)}$ weakly converging to some $Q^{(\infty)}\in \Delta\big([0,1]^N\big)$.

The conditions of Lemma~\ref{lem:domination} can be rewritten in an equivalent integrated form. Condition~$(1)$ becomes 
\begin{equation}\label{eq_weak_domination}
\int f(x_1,\ldots,x_n)\left(\frac{1}{p}\dd\p(x_1,\ldots,x_n)-\dd Q(x_1,\ldots,x_n)\right)\geq 0
\end{equation}
for any non-negative continuous function $f$ on the unit cube. Condition~$(2)$ is equivalent to 
\begin{equation}\label{eq_weak_marginals}
\int g(x_i)\left(\frac{x_i}{p}\dd\p(x_1,\ldots,x_n)-\dd Q(x_1,\ldots,x_n)\right)= 0
\end{equation}
for any agent $i\in N$ and any continuous $g$ of arbitrary sign on the unit interval.

With this reformulation it is immediate that both conditions are closed, and hence withstand the weak* limit. Therefore, since $\p^{(k_m)}$, $ Q^{(k_m)}$ satisfy the conditions~\eqref{eq_weak_domination} and~\eqref{eq_weak_marginals}, the limiting pair $\p^{(\infty)}$, $ Q^{(\infty)}$ satisfies them as well. We deduce that  $\p^{(\infty)}$ is feasible. \end{proof}

\section{Checking on Intervals is not Sufficient}
\label{sec:intervals}
In this section we show that restricting the condition \eqref{eq:aumann-feasible-quant-2} to intervals $A_1,A_2$ provides a condition that---while clearly necessary---is not sufficient for feasibility.

To see this, consider the following distribution $\p$, depicted in Figure~\ref{figure:intervals}. It is parameterized by small $\varepsilon>0$ and is supported on four points: two ``heavy'' points, $\left(\frac{1}{4}-\frac{\varepsilon}{4},\frac{1}{2}\right)$ and $\left(\frac{3}{4},\frac{1}{2}\right)$ with probabilities $\frac{1-\varepsilon}{2}$ and $\half$, respectively, and two ``light'' points, $\left(\frac{1}{2}-\frac{\varepsilon}{4},0\right)$ and $\left(\frac{1}{2}-\frac{\varepsilon}{4},1\right)$, each with probability $\frac{\varepsilon}{4}$. Thus
\begin{align}
    \label{eq:intervals}
\p =
\frac{1-\varepsilon}{2}\delta_{\frac{1-\varepsilon}{4},\frac{1}{2}} +
\frac{1}{2} \delta_{\frac{3}{4},\frac{1}{2}} + \frac{\varepsilon}{4}\left(\delta_{1/2-\varepsilon/4,0}+
\delta_{1/2-\varepsilon/4,1}\right)
\end{align}

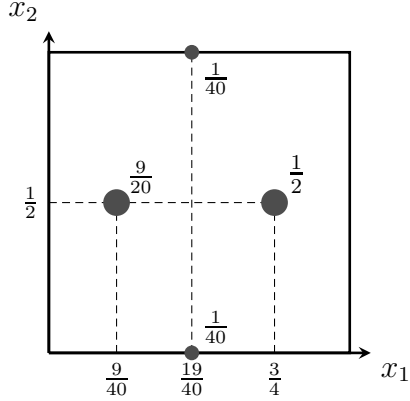
\begin{figure}[h]
\begin{center}
\begin{tikzpicture}[scale=0.4, line width = 1pt]
\draw (0,0)--(10,0)--(10,10)--(0,10)--(0,0);
		\draw [->,>=stealth] (0,0) -- (10.7,0);
\draw [->,>=stealth] (0,0) -- (0,10.7);
\node[below right] at (10.7,0) {\large $x_1$};
\node[above left] at (0,10.7) {\large $x_2$};


\draw [thin, densely dashed] (2.25,0) -- (2.25,5);
\node[below] at (2.25,0) {$\frac{9}{40}$};
\draw [thin, densely dashed] (7.5,0) -- (7.5,5);
\node[below] at (7.5,0) {$\frac{3}{4}$};

\node[left] at (0,5) {$\frac{1}{2}$};
\draw [thin, densely dashed] (0,5) -- (7.5,5);

\filldraw[black!70] (2.25,5) circle (0.4);
\node[above right] at (2.26,5) {$\frac{9}{20}$};
\filldraw[black!70] (7.5,5) circle (0.4);
\node[below] at (4.75,0) {$\frac{19}{40}$};
\draw [thin, densely dashed] (4.75,0) -- (4.75,10);

\node[above right] at (7.6,5) {\large $\frac{1}{2}$};
\filldraw[black!70] (4.75,0) circle (0.2);
\node[above right] at (4.75,0) {$\frac{1}{40}$};
\filldraw[black!70] (4.75,10) circle (0.2);
\node[below right] at (4.75,10) {$\frac{1}{40}$};
\end{tikzpicture}
\end{center}
\caption{The distribution defined in \eqref{eq:intervals}, for $\varepsilon=1/10$. The ``heavy'' points are large, and the ``light'' points are small.\label{figure:intervals}}
\end{figure}

This distribution satisfies the martingale condition with prior $\half$; however, it is not $\half$-feasible. To see that, pick $A_1=\left\{\frac{1-\varepsilon}{4}, \frac{3}{4}\right\}$ and $A_2=\left\{\frac{1}{2}\right\}$, then $A_1\times \overline{A_2}$ has zero $\p$-measure and the condition \eqref{eq:aumann-feasible-quant-2} is violated:
$$0\geq \frac{1-\varepsilon}{4}\cdot\frac{1-\varepsilon}{2}+\frac{3}{4}\cdot\frac{1}{2}-\frac{1}{2}\left(\frac{1-\varepsilon}{2}+\frac{1}{2} \right)=\frac{\varepsilon^2}{8}.$$

We now check that none of the inequalities \eqref{eq:aumann-feasible-quant-2} is violated for intervals $A_1, A_2$. Since $\p$ satisfies the martingale condition, it suffices to check the left inequality\footnote{Using the martingale condition $\int_{A_1} x \,\dd \p_1(x) + \int_{\overline{A_1}}x\,\dd \p_1(x)=p$ the right inequality in \eqref{eq:aumann-feasible-quant-2} follows from the left by a simple calculation.}
\begin{align}
    \label{eq:aumann-feasible-quant-2-left}
  \p(A_1 \times \overline{A_2}) \geq \int_{A_1} x\,\dd \p_1(x) - \int_{A_2} x\,\dd \p_2(x). 
\end{align}
Since $\p$ has finite support, different choices of $A_1,A_2$ yield the same inequality if each set contains the same points of the support of $\p_1$ and $\p_2$, respectively. Thus, we need to check that \eqref{eq:aumann-feasible-quant-2-left} is satisfied if $A_1$ is a subset of $\left\{\frac{1-\varepsilon}{4},\, \frac{1}{2}-\frac{\varepsilon}{4},\,\frac{3}{4}\right\}$ and $A_2$, a subset of $\left\{0,\,\frac{1}{2},\,1\right\}$, {\em except} the cases of $A_1=\left\{\frac{1-\varepsilon}{4},\,\frac{3}{4}\right\}$ and $A_2 = \left\{0,\,1\right\}$, which exclude the middle points and do not correspond to any interval.

We consider the following cases:
\begin{itemize}
\item Inequality \eqref{eq:aumann-feasible-quant-2-left} holds if one of the sets $A_i$ is empty or contains the whole support of $P_i$, as it then boils down to the martingale condition, which we already verified.
\item Consider the case when $A_1\times \overline{A_2}$ contains exactly one heavy point; by the interval constraint, if it contains two, then $A_1$ contains the support of $\p_1$, which is the case we already considered. In this case, $\p(A_1\times \overline{A_2})\geq \frac{1-\varepsilon}{2}$. On the other hand, the integral $\int_{A_1} x\,\dd\p_1(x)$ does not exceed $\frac{1}{2}-\frac{3}{4}\cdot\frac{1}{2}=\frac{1}{8}$ if $A_1$ excludes the rightmost heavy point and it does not exceed $\frac{1}{2}-\frac{1-\varepsilon}{4}\frac{1-\varepsilon}{2}$. We see that for $\varepsilon$ small enough (e.g., $\varepsilon=\frac{1}{10}$),  condition~\eqref{eq:aumann-feasible-quant-2-left} is satisfied regardless of the choice of $A_2$, since $\frac{1-\varepsilon}{2}\geq \max\left\{\frac{1}{8},\ \frac{1}{2}-\frac{1-\varepsilon}{4}\frac{1-\varepsilon}{2}\right\}$.
\item  Consider the remaining case, in which there are no heavy points in $A_1\times \overline{A_2}$, and both $A_1$ and $A_2$ are nonempty strict subsets of the supports. This can only be possible if $A_2$ contains $\half$ or  $A_1=\left\{\frac{1}{2}-\frac{\varepsilon}{4}\right\}$. In the former case, $\int_{A_2}x\,\dd\p_2(x)\geq \frac{1}{2}-\frac{\varepsilon}{4}$, which for small $\varepsilon$ ($\frac{1}{10}$ suffices) exceeds $\int_{A_1} x\,\dd \p_1(x)$ for any $A_1$ excluding at least one of the heavy points (see the bounds above). Hence, the right-hand side of \eqref{eq:aumann-feasible-quant-2-left} is negative and the inequality is satisfied. Consider the remaining case of $A_1=\left\{\frac{1}{2}-\frac{\varepsilon}{4}\right\}$ and $A_2$ either $\{0\}$ or $\{1\}$. The left-hand side of \eqref{eq:aumann-feasible-quant-2-left} equals $\frac{\varepsilon}{4}$ and $\int_{A_1} x\, \dd\p_1(x)=\left(\frac{1}{2}-\frac{\varepsilon}{4}\right)\cdot\frac{\varepsilon}{2}\leq \frac{\varepsilon}{4}$. Thus \eqref{eq:aumann-feasible-quant-2-left} is satisfied on all intervals.

\end{itemize}

This example also demonstrates that the condition of \cite{ziegler2020adversarial} (Theorem 1) is not sufficient for feasibility. 
Using our notation his condition can be written as follows:
\begin{align*}
  &\max\left\{
\int_0^a x\,d\p_1(x)+ \int_0^b x\,d\p_2(x)-p, \ \ \
\int_0^a (1-x)\,d\p_1(x)+ \int_0^b (1-x)\,d\p_2(x)-(1-p)           \right\} \leq  \\
&\leq \p([0,a]\times [0,b])\leq \\
&\leq\min\left\{
\int_0^a (1-x)\,d\p_1(x)+ \int_0^b x\,d\p_2(x), \ \ \
\int_0^a x\,d\p_1(x)+ \int_0^b (1-x)\,d\p_2(x)\right\}. \\
&\text{for every } a,b\in [0,1]
\end{align*}
Simple computations show that for every given $a,b\in [0,1]$ this condition is \emph{equivalent} to four of our conditions (see Equation \ref{eq:aumann-feasible-quant-2}) for all possible combinations of $A_1=[0,a]$ or $A_1=[a,1]$ and $A_2=[0,b]$ or $A_2=[b,1]$. 

The condition for feasibility in \cite{ziegler2020adversarial} is essentially the same condition as our condition \eqref{eq:aumann-feasible-quant-2}, but restricted to interval subsets $A_1,A_2 \subset [0,1]$, where each interval has a boundary point either at 0 or at 1. If the marginals $\p_1$ and $\p_2$ are supported on at most two points, such intervals exhaust all possible non-trivial sets in Equation \ref{eq:aumann-feasible-quant-2};  in this case, Ziegler's condition is necessary and sufficient for feasibility. However, the example above demonstrates that this condition becomes insufficient if the support contains at least $3$ points.

\section{Singularity of extreme measures with given marginals}\label{sec_Linden}

In this section we formulate an extension of a celebrated result of \citet{lindenstrauss1965remark} regarding extreme points of the set of measures with given marginals. The extension is two-fold: the classic result assumes $n=2$ and uniform marginals and we get rid of both assumptions. This extension can also be deduced from a more general statement by \cite{shortt1986singularity}, which allows for non-orthogonal multidimensional projections.

For the reader's convenience, we include a proof similar to Lindenstrauss's.
\begin{theorem}[$n=2$ \citep{lindenstrauss1965remark}; $n\geq 2$ \citep{shortt1986singularity}]\label{th_Linden}
Any extreme point $\p$ of the set of probability measures on $[0,1]^N$ with given one-dimensional marginals $\nu_i\in \Delta([0,1]),$ $i\in N$, is singular with respect to the Lebesgue measure $\lambda$ on the unit cube. In other words, for each extreme $\p$ there exists a measurable set $B$ such that $\p(B)=1$ and $\lambda(B)=0$.
\end{theorem}	
\begin{proof}
Assume the converse: $\p$ is not singular. By the Lebesgue decomposition theorem the continuous part of $\p$ can be singled out: $\p=\mu+\lambda^{\perp}$, where $\mu\ne 0$ is absolutely continuous with respect to $\lambda$ ($d\mu=f\dd\lambda$ with a non-negative integrable density $f$), $\lambda^\perp$ is singular with respect to $\lambda$. 


Consider the Lebesgue space $\cL^1(\mu)$ of integrable functions with respect to $\mu$ (defined $\mu$-everywhere) and its closed subspace $\mathcal{S}$ generated by ``separable''\footnote{We refer here to separability in economic sense.} functions $g(x_1,\ldots,x_n)=\sum_{i\in N} g_i(x_i)$, where $g_i \in \cL^1(\mu)$, $i\in N$, depends on the variable $x_i$ only.

Let's assume that $\mathcal{S}\ne \cL^1(\mu)$, i.e., separable integrable functions are not dense in all integrable (we check this condition at the end of the proof). Now we show that under this assumption, $\mu$ can be represented as a convex combination of $\frac{1}{2}\mu'+\frac{1}{2}\mu''$, where $\mu'$ and $\mu''$ are distinct but have the same marginals. By the Hann-Banach theorem there exists a continuous functional $\theta$ of norm $1$ such that $\theta$ is identically zero on $\mathcal{S}$. Since the dual space to  $\cL^1(\mu)$ is the space $\cL^\infty(\mu)$ of essentially-bounded functions, the functional $\theta$ can be identified with a non-zero function $\theta(x_1,\ldots,x_n)$ bounded by $1$ in absolute value. The condition of vanishing on $\mathcal{S}$ reads as
\begin{equation}
\int  \sum_{i\in N} g_i(x_i) \cdot \theta (x_1,\ldots,x_n)\, \dd\mu=0,\qquad \forall g_i=g_i(x_i)\in \cL^1(\mu), \ \ i\in N.
\end{equation}
We see that the measure $\theta\,\dd\mu$ is non-zero but has zero marginals! Define $\mu'$ and $\mu''$ as  $d\mu'=\left(1-\theta\right)\dd\mu$ and $d\mu'=\left(1+\theta\right)\dd\mu$. By the construction, $\mu$ is the average of $\mu'$ and $\mu''$, $\mu'$ and $\mu''$ are distinct, and all measures $\mu$, $\mu'$, and $\mu''$ have the same marginals.  Thus $\p$ is also represented as the average of $\mu'+\lambda^\perp$ and $\mu''+\lambda^\perp$, i.e., $\p$ is not an extreme point. Contradiction. This contradiction completes the proof (under the assumption that $\mathcal{S}\ne \cL^1(\mu)$).

Now we check that $\mathcal{S}\ne \cL^1(\mu)$. Our goal is to construct a function $h\in \cL^1(\mu)$ such that for any $g(x_1,\ldots,x_n)=\sum_{i\in N} g_i(x_i)\in \cL^1(\mu)$ the $\cL^1$-distance $\int |h-g|\,\dd\mu$ is bounded below by some positive constant independent of $g$. 

Recall that $\dd\mu=f\dd\lambda$. Fix $\delta>0$ such that the set $A_\delta=\{x\in [0,1]^N \, : \, f(x)\geq \delta\}$ has non-zero Lebesgue measure. Fix another small constant $\varepsilon>0$; by the Lebesgue density theorem applied to $A_\delta$, there exist a point $x^0\in (0,1)^N$ and a number $a>0$ such that for the cube $C=\prod_{i\in N} [x_i^0-a,x_i^0+a)\subset [0,1]^N$ of size $2a$ centered at $x^0$ the following inequality holds $\frac{\lambda(C \setminus A_\delta)}{\lambda(C)}\leq \varepsilon$. 

Define $h(x)=1$ if $x_i\geq x_i^0$ for all $i\in N$, and $h(x)=0$, otherwise. 

Cut the cube $C$ into $2^n$ small cubes indexed by subsets of $N$: for $M\subset N$ the cube $C_M$ is given by  $\prod_{i\in N\setminus M} \big[x_i^0-a,\,x_i^0\big)\times \prod_{i\in M} \big[x_i^0,\,x_i^0+a\big)$. No function $g(x)=\sum_{i\in N} g_i(x_i)$ can approximate $h$ well on all the small cubes at the same time. Intuition is the following: since $h$ is zero in all the cubes except $C_N$, then $g$ must be close to zero on these cubes; however, values on these cubes determine values of $g$ on $C_N$ by
\begin{equation}\label{eq_g_identity}
g(x)=\frac{1}{n-1}\left(\left(\sum_{i\in N} g\big(x_1,\ldots,x_{i-1},\,x_i-a,\,x_{i+1},\ldots,x_n\big)\right)-g\big(x_1-a,\,x_2-a,\ldots,x_n-a\big)\right),
\end{equation}
therefore $g$ is close to zero on $C_N$ and cannot approximate $h$ well.

To formalize this intuition, we assume that $\int_{C_M}|h-g|\dd\mu$ is less then some $\alpha\cdot \lambda(C_M)$ for any $M\subset N$ and show that this constant $\alpha$ cannot be too small. For  $M\ne N$ we get
$\int_{C_M} |g|\cdot f\,\dd\lambda \leq \alpha\lambda(C_M)$. Applying the Markov inequality on the set $C_M\cap A_\delta$ and taking into account that this set is big enough (by the construction of the original cube, $\frac{\lambda(C_M \setminus A_\delta)}{\lambda(C_M)}\leq 2^n\varepsilon$), we obtain existence of a set $B_M\subset C_M$ such that $|g(x)|\leq \frac{\sqrt{\alpha}}{\delta}$ on $B_M$ and 
$\frac{\lambda( C_M \setminus B_M)}{\lambda(C_M)}\leq 2^n\varepsilon+\sqrt{\alpha}$.

Consider a subset $B^*$ of $C_N$ such that, whenever $x\in B^*$, all the arguments of the right-hand side in~\eqref{eq_g_identity} belong to respective subsets $B_M$, i.e., $B^*=\cap_{i\in N} \left(B_{N\setminus\{i\}}+a\cdot {e_i}\right)\cap \left(B_{\emptyset}+a\cdot\sum_{i\in N} {e_i}\right)$, where sets $B_M$ are translated by the elements of the standard basis $(e_i)_{i\in N}$. The union bound implies that the set $B^*$ is dense enough in $C_N$: $\frac{\lambda( C_N \setminus B^*)}{\lambda(C_N)}\leq (n+1)\left(2^n\varepsilon+\sqrt{\alpha}\right)$. By formula~\eqref{eq_g_identity}, the absolute value of $g$ is bounded by $\frac{n}{n-1}\frac{\sqrt{\alpha}}{\delta}$  on $B^*$. We get the following chain of inequalities: 
\begin{align} 
\label{eq_chain_lambda}
\alpha 
&\geq  \frac{1}{\lambda(C_N)}\int_{C_N}|h-g|\,\dd\mu \nonumber\\
&\geq \frac{1}{\lambda(C_N)}\int_{B^*\cap A_\delta} |h-g|\cdot f\,\dd\lambda \nonumber\\
&\geq \left(1-\frac{n}{n-1}\frac{\sqrt{\alpha}}{\delta}\right)\cdot \delta\cdot  \frac{\lambda(B^*\cap A_\delta)}{\lambda(C_N)} \nonumber\\ \noindent
&\geq\left(\delta-\frac{n}{n-1}{\sqrt{\alpha}}\right)\left(1-(n+1)\left(2^n\varepsilon+\sqrt{\alpha}\right)-2^n\varepsilon\right).
\end{align}
Denote by $\alpha^*=\alpha^*(\delta, n,\varepsilon)$ the minimal value of $\alpha\geq 0$ satisfying the inequality created by the head and the tail of~\eqref{eq_chain_lambda}. The constant $\varepsilon$ is a free parameter in the construction. Selecting it to be small enough  (namely $(n+2)2^n\varepsilon< 1$), we ensure that  $\alpha^*>0$. 

For any $g=\sum_{i\in N} g_i(x_i)$
$$\int_{[0,1]^N} |h-g|\,\dd\mu\geq \max_{M\subset N}\int_{C_M} |h-g|\,\dd\mu\geq \alpha^*(\delta,n,\varepsilon)\cdot 2^{-n}\lambda(C)>0.$$
Since the constant on the right-hand side is independent of $g$ and positive, we see that $h$ cannot be approximated by separable functions. Thus $h$ does not belong to $\mathcal{S}$ and $\mathcal{S}\ne\cL^1$.
\end{proof}

\section{Independent beliefs induced by Gaussian signals}\label{sect_Gaussian}
Let $\phi$ be the density of the standard Gaussian random variable: $\phi(t)=\frac{1}{\sqrt{2\pi}}e^{-\frac{t^2}{2}}$.

For prior $p=\half$, consider an agent who gets a signal $s\in\R$ distributed according to the Gaussian distribution with variance $1$ and mean equal to $d$ for the state $\omega=h$, and equal to $-d$ for the state $\omega=\ell$, so that the conditional distributions have the densities $\phi(s-d)$ and $\phi(s+d)$, respectively. The density $f(s)$ of the unconditional distribution is $$f(s)=\frac{1}{2}\phi(s-d)+\frac{1}{2}\phi(s+d)=\frac{1}{2\sqrt{2\pi}}e^{-\frac{d^2}{2}}e^{-\frac{s^2}{2}}\big(e^{ds}+e^{-ds}\big).$$ For the sake of definiteness we assume $d>0$.

Denote by $\nu$ the induced distribution of posteriors. By Bayes' Law, the posterior $x(s)$ upon receiving the signal $s$ is equal to 
\begin{align}
\label{eq:xs}
x(s)=\frac{\phi(s-d)}{\frac{1}{2}\phi(s-d)+\frac{1}{2}\phi(s+d)}\cdot \frac{1}{2}=\frac{e^{ds}}{e^{ds}+e^{-ds}}    
\end{align}
and 
$$
  \nu([0,t])=\int_{-\infty}^{x^{-1}(t)} f(s)\dd s.
$$

We are interested in the question of when $\nu\times\nu$ is $\half$-feasible. By Proposition~\ref{thm:product-feasible}, it is feasible if and only if it is a mean preserving contraction of the uniform distribution. The next lemma provides a simple sufficient condition for this property. 
\begin{lemma}\label{lm_simple_criterion_of_product_feasibility}
Let $\mu\in\Delta([0,1])$ be a non-atomic distribution, symmetric around $\half$, with cumulative distribution function $F(a) = \mu([0,a])$. Assume that there is at most one point $\frac{1}{2}<a<1$ such that $F(a)=a$ and that $F(x)-x>0$ for all $x$ close enough to $1$, but not equal to 1. Then the uniform distribution is a mean preserving spread of $\mu$  if and only if
\begin{equation}\label{eq_condition_for_feasibility_a_equals_1}
\frac{1}{8}\geq \int_{\frac{1}{2}}^1 \left(x-\frac{1}{2}\right)\, \dd \mu(x).
\end{equation}
\end{lemma}
\begin{proof}
Necessity follows directly from the convex order definition of mean preserving spreads. Indeed, the condition for being a mean preserving contraction of the uniform distribution  is equivalent to  
\begin{equation}\label{eq_2nd_order}
 H(y)=\int_0^y (F(x)-x)\,\dd x\leq 0\quad\quad \forall y\in [0,1].   
\end{equation} 
Integration by parts implies that for $a=\half$ this condition becomes exactly~\eqref{eq_condition_for_feasibility_a_equals_1}.

To prove sufficiency, we must check the inequality~\eqref{eq_2nd_order} for all $y\in [0,1]$. Symmetry of $\mu$ implies $H(y)=H(1-y)$ and thus we can focus on the right sub-interval $y\in [\half,1]$. It is enough to check that $\max_{y\in [\half,1] } H(y)\leq 0$. We claim that the maximum is attained at one of the end-points. Indeed, if there is an internal maximum $a\in (\half,1)$, then  the derivative $H'(a)=F(a)-a=0$. On the whole interval $(a,1)$, the derivative $H'(y)$ has a constant sign since it is continuous and  $y=a$ is the unique zero in $\left(\frac{1}{2},1\right)$.
By assumption, $F(x)-x>0$ for $x$ close to $1$
and, hence, the derivative is positive on $(a,1)$ and $H$ is increasing on this interval. Therefore, the point $a$ cannot be a maximum of $H$ on $[\half,1]$.
Thus~\eqref{eq_2nd_order} holds if and only if it holds on the end-points $y=\half$ and $y=1$.
For  $y=\half$, this inequality coincides with~\eqref{eq_condition_for_feasibility_a_equals_1} and for $y=1$ it is trivial.
\end{proof}

Below we check that $\nu$, as induced by Gaussian signals, satisfies the assumptions of Lemma~\ref{lm_simple_criterion_of_product_feasibility} and, therefore, $\nu\times \nu$ is $\half$-feasible if and only if 
$\int_{\frac{1}{2}}^1 \left(x-\frac{1}{2}\right)\,\dd \nu \leq \frac{1}{8}$. This is equivalent to $\int_{\frac{1}{2}}^1  x\,\dd \nu \leq \frac{3}{8}$.

We can make this condition more explicit by rewriting its left-hand side as 
$$\int_{\frac{1}{2}}^1 x\,\dd \nu(x) = \int_0^\infty x(s)\cdot f(s) \,\dd s= \frac{1}{2}\int_0^\infty \phi(s-d)\,\dd s= \frac{1}{2}\int_{-\infty}^d \phi(t)\, \dd t,$$
where in the last equation we applied the change of variable $s-d=-t$.  This results in the following condition of feasibility:
$$\int_{-\infty}^{d} \phi(t)\, \dd t\leq  \frac{3}{4}, $$
i.e., $d$ must be  below the $\frac{3}{4}$-quantile of the standard normal distribution.\footnote{Note that in this computation, we do not use the explicit formula for $\phi$. In particular, one gets the same condition of feasibility for any absolutely-continuous distribution of signals on $\R$, if the induced distribution of posteriors $\nu$ satisfies the assumptions of Lemma~\ref{lm_simple_criterion_of_product_feasibility}.}

It remains to prove that $\nu$ indeed satisfies the assumptions of Lemma~\ref{lm_simple_criterion_of_product_feasibility}. We check that $\nu([0,x])-x$ is positive for $x$ close to $1$. Recall that we can write the belief $x$ as a function $x(s)$ of the signal $s$ using \eqref{eq:xs}. We denote the derivative of $x(s)$ by $x'(s)$. Substituting $x=x(s)$, we get
$$\nu([0,x(s)])-x(s)=1-x(s)-\nu([x(s),+\infty))= \int_{s}^\infty \left(x'(t)-f(t)\right)\dd t.$$
For large $t$, we have 
$$
  f(t)=\frac{1}{2\sqrt{2\pi}}e^{-\frac{d^2}{2}}\cdot e^{-\frac{t^2}{2}+d\cdot t}(1+o(1))\ \ \ \ \mbox{and} \ \ \ \ x'(t)=\frac{2d}{\big(e^{d\cdot t}+e^{-d \cdot t}\big)^2}=2d\cdot e^{-2d \cdot t}(1+o(1)).$$
Therefore,  the asymptotic behavior of the integrand is dictated by $x'(t)$:
$$x'(t)-f(t)=x'(t)(1+o(1)).$$
The integrand is positive for $t$ large enough. This implies the desired positivity of $\nu([0,x(s)])-x(s)$ for large values of $s$.

Now we check that there is at most one point $a\in\left(\frac{1}{2},1\right)$ such that $F(a)=a$ or, equivalently, there is at most one point $s\in(0,\infty)$ such that
$\int_{-\infty}^s f(t)\,\dd t- x(s)=0$. Denote $G(s)=\int_{-\infty}^s f(t)\,\dd t- x(s)$. The function $G$ is smooth and $G(0)=\lim_{s\to+\infty} G(s)=0$. If $G$ has $k$ zeros in $(0,\infty)$, then it also has at least $k+1$ extrema (minima or maxima) in this interval, hence, at least $k+1$ critical points (zeros of the derivative $G'$). We will show that there are no more than two critical points and thus $G$ has at most one zero. The equation for critical points takes the following form
$$G'(s)=0 \ \Longleftrightarrow \ \frac{1}{2\sqrt{2\pi}}e^{-\frac{d^2}{2}}e^{-\frac{s^2}{2}}\big(e^{d\cdot s}+e^{-d \cdot s}\big)-\frac{2d}{\big(e^{d\cdot s}+e^{-d\cdot s }\big)^2}=0$$
and can be rewritten as
$$e^{-\frac{s^2}{2}}\big(e^{d \cdot s}+e^{-d\cdot s}\big)^3=4d\sqrt{2\pi}e^{\frac{d^2}{2}}.$$
Denote the left-hand-side by $H(s)$. The graph of $H$ for $s\geq 0$ can intersect any given level at most twice since $H'$ has at most one zero in $(0,+\infty)$. Indeed, $H'(s)=\left(s-3d\cdot \mathrm{tanh}(dq) \right)\cdot H(q)$ and the equation $q-3d\cdot \mathrm{tanh}(dq)=0$ has at most one positive solution by concavity of the hyperbolic tangent on $[0,+\infty)$.
Therefore, $G$ has at most $2$ critical points and thus at most one zero in $(0,\infty)$, which completes the argument and justifies the application of Lemma~\ref{lm_simple_criterion_of_product_feasibility} to $\nu$.

\end{document}